\documentclass[12pt, aps, prb, preprint]{revtex4}
 \usepackage[T1]{fontenc}
 \usepackage[cp1250]{inputenc}
 \usepackage{amssymb,amsfonts,amsmath,amsthm}
 \usepackage{graphics,graphicx}
\newtheorem{theor}{Theorem}

\newcommand{\M}{\mathcal{M}}
\newcommand{\ve}{\overrightarrow}
\newcommand{\vep}{\varepsilon}
\newcommand{\lp}{\left(}
\newcommand{\rp}{\right)}
\newcommand{\pr}{\mathrm{pr}}

\begin{document}
\vspace*{1cm}
\title{Flat spacetime in a~capsule\vspace{1cm}}

\author{Andrzej Herdegen}
\affiliation{Institute of Physics, Jagiellonian University,
Reymonta 4, 30-059 Cracow, Poland\vspace{2cm}}
\email{herdegen@th.if.uj.edu.pl}
%\date{}

\begin{abstract}
 We propose a~parallel introduction to Galilean and Einsteinian
 relativity based on the causal structure and inertial motions.
 Galilean and Poincar\'e transformations, as objects secondary
 to the geometrical structure, are left aside.
\end{abstract}

\maketitle

\section{Introduction}

This article is intended for university level teachers
lecturing, and students learning, special relativity (SR). It
is not meant as a~text which could be directly used as a~SR
primer. Rather it gives a~background, or an outline on which
one can elaborate the exposition of SR. We assume as
a~background for this article a~course in linear algebra
including real vector and affine spaces, direct sums of
subspaces, linear forms and symmetric bilinear forms of any
signature.

We propose a~highly structured and logical approach to the
fundamentals of SR based on its causal structure and relativity
of inertial motions. For comparison and better understanding we
parallelly build the Galilean spacetime (GS) on similar ideas.
We indicate that the causal structure determines the metric
structure of SR spacetime uniquely, which is not the case for
the choice of Euclidean metric in the Galilean case.

We want to stress the point that the Galilean and Lorentz
(Poincar\'e) transformations are objects secondary to the
geometric structure of spacetime: they are affine mappings
leaving this structure invariant. We regard basing the
introduction to SR on these transformations as a~serious
misconception and we do not discuss them in this article.

We are also of the opinion that introducing SR, for the sake of
alleged simplicity, from the three-dimensional rather than full
geometrical point of view, in fact makes understanding of SR
more difficult, and can easily lead to misconceptions. We
regard as especially harmful figures illustrating hypothetical
relative motion of frames as depicted in Fig.~1. Whereas this
is not the best, but correct picture in GS, it is completely
wrong in SR. The reason for that is that the hyperplanes of
constant time (`pure space') of observers in relative motion
are not parallel, so they cannot be regarded as `sliding' on
each other.

Elements of the programme sketched above appeared, of course,
in many earlier publications and books (see e.g.~Refs.~1-3) but
we believe that our scheme adds some value to the clarity and
logic.

In addition we discuss some simple geometric effects in the
present context. This will include a~discussion of the view of
the celestial sphere as seen by different
observers.\cite{pen,ter} This point is particularly worth
adding, as it is usually treated with the help of a~rather
indirect method of stereographic projection.\cite{cel} We
discuss it directly on the celestial spheres of two observers.

In all discussions of effects involving different observers we
consistently avoid, as mentioned above, the use of Galilean or
Lorentz transformations. To relate the views on the spacetime
as seen by two inertial observers one needs only to know the
directional vectors of their world-lines. On the other hand one
needs complete bases attached to the observers to specify a
transformation between them.

\section{Homogeneity with respect to translations and the affine
structure}\label{affine}

It is fairly obvious from everyday experience that one needs
four real numbers to place an event in space and time. For
a~given event the specific values of these numbers depend on an
adopted system of labels, but they always form an element of
the set $\mathbb{R}^4$. Our spacetime is a~structure based on
this set.

Another common experience points to the applicability of
spacetime translations: if a~physical occurrence takes place in
a~given region of space and within some time-span, an analogous
occurrence may take place elsewhere and at another time. We
include this property in our construction of a~model of the
spacetime in the following form: the group of four-dimensional
translations acts transitively on the spacetime. This leads us
to the following starting point for the construction of
a~spacetime model:
\begin{center}
Flat spacetime is modeled by a~real four-dimensional affine
space~$(\M,M)$.
\end{center}

Here $\M$ denotes the affine space based on the
four-dimensional vector space~$M$. We adopt the notation
$P,Q,\ldots$ for points in $\M$ and $x,y,\ldots$ for vectors in
$M$. We write $x=\ve{PQ}$ if $Q=P+x$. Moreover, if $P\in\M$ and
$N\subset M$ is any subset then we use the usual shorthand:
$P+N= \{P+x\mid x\in N\}$. In particular, straight lines are
one-dimensional affine subspaces $P+L(x)$, where $L(x)$ denotes
the one-dimensional vector subspace spanned by the vector~$x$.
Ordered vector bases in $M$ will be denoted by
$(e_0,e_1,e_2,e_3)$. See Fig.~2 for a~graphic representation
(here, as in the following, one space dimension is omitted).

\section{Causal structure and inertial motions}\label{causal}

Of course, the affine space structure is still a~very poor one,
one needs further specification. The most obvious element
needed is a~one introducing the differentiation between
physical time and space directions. This is achieved in the
following way.

We shall say that the spacetime is equipped with the
\emph{causal structure} if in the accompanying vector space one
has distinguished the following set (see Fig.~3):
\begin{itemize}
 \item[GS:] a~three-dimensional subspace $S\subset M$,
 \item[SR:] a~homogeneous vector quadric $V\subset M$
     (different from a~subspace), with respect to which
     three dimensions of $M$ are on equal footing, but not
     the fourth.
\end{itemize}

By a~homogeneous vector quadric we mean here a~set of vectors
$x\in M$ whose coordinates $x^0,x^1,x^2,x^3$ in some (and then
any) basis satisfy the equation
\mbox{$\sum_{i,j=0}^3\alpha_{ij}x^ix^j=0$}, with some
basis-dependent numerical coefficients $\alpha_{ij}$. We recall
that for any such quadric there is a~basis in which it takes
one of the forms
$\vep_0(x^0)^2+\vep_1(x^1)^2+\vep_2(x^2)^2+\vep_3(x^3)^2=0$,
where $\vep_\mu=0,\pm1$ (uncorrelated values). The only
possibility (up to a~permutation of the basis vectors) to
satisfy the demand imposed above on $V$ is that in this
canonical basis $V$ is a~cone given by:
\begin{equation}\label{cone}
 x\in V\quad \iff\quad (x^0)^2-(x^1)^2-(x^2)^2-(x^3)^2=0\,.
\end{equation}
We shall say that a~vector $x$ lies inside (or outside) $V$ if
\mbox{$(x^0)^2-(x^1)^2-(x^2)^2-(x^3)^2>0$} ($<0$) respectively.

We say that a~nonzero vector is a~\emph{causal vector} if it:
\begin{itemize}
 \item[GS:]  does not lie in $S$,
 \item[SR:] lies inside or on $V$.
\end{itemize}
In addition we introduce the notion of a~\emph{timelike vector}
which
\begin{itemize}
 \item[GS:] is identical with a~causal vector,
 \item[SR:] lies inside $V$.
\end{itemize}

We shall say that two events $P$ and $Q$ are \emph{causally
related} if $\ve{PQ}$ is a~causal vector, and they are
\emph{temporally related} if it is a~timelike vector.

The causal structure makes contact with physics by the
following identification. An \emph{inertial motion} is
a~straight line in spacetime $\M$ with a~timelike directional
vector (thus any two events on this line are temporally
related). Such lines will be called world-lines of the motion
(see Fig.~4)

If a~point $Q\neq P$ is not causally related to $P$ we say that
it lies \emph{elsewhere} with respect to $P$. One then cannot
reach $Q$ from $P$ by an inertial motion.

\section{The four orientations of the
spacetime}\label{orientation}

Let us choose a~basis of $M$ in which
\begin{itemize}
 \item[GS:] the subspace $S$ is given by $x^0=0$,
 \item[SR:] the cone $V$ has the canonical form.
\end{itemize}
The set of causal vectors splits into two disjoint sets: those
for which $x^0>0$ or $x^0<0$ respectively in the distinguished
basis. We denote one of these sets by $C_+$ and call it the
future and the other by $C_-$ and call it the past. (After this
choice has been done we can adjust the sign of $x^0$ so that
$x^0>0$ for $x\in C_+$.) Then the future (past) of any event
$P$ is the set $P+C_+$ ($P+C_-$), and $Q$ is in the future of
$P$ if, and only if, $P$ is in the past of $Q$. Let us write
$Q>P$ for ``$Q$ is in the future of $P$'', and $Q\geq P$ for:
$Q>P$ or $Q=P$. Then the relation $Q\geq P$ defines a~partial
order in $\M$:
\begin{itemize}
 \item[$1^\mathrm{o}$] $P\geq P$,
 \item[$2^\mathrm{o}$] if $Q\geq P$ and $P\geq Q$ then
     $Q=P$,
 \item[$3^\mathrm{o}$] if $R\geq Q$ and $Q\geq P$ then
     $R\geq P$.
\end{itemize}
The only less obvious of these properties is the third one in
the special relativity case. To prove it observe that $x\in
C_+$ if in a~canonical basis
\mbox{$0<x^0\geq\sqrt{(x^1)^2+(x^2)^2+x^3)^2}$}. If $y$ is
another such vector then it is easily seen that the same
relation is satisfied with $x$ replaced by $x+y$, which was to
be proved. See Fig.~5 for a~graphic representation of causally
defined regions.

As there are two possible choices for the identification of the
sets $C_\pm$ we say that there are two possible \emph{causal
orientations} of the spacetime $\M$.

At the same time $M$ as a~real vector space has two possible
orientations defined as usually as the equivalence classes of
bases. In combination with the causal orientation this gives
four choices of the spacetime $\M$ orientations.

\section{Relative rest, inertial observers, inertial
frames}\label{frames}

We do not have yet any metric tools, so we are unable to
determine relative velocity of inertial motions, but we can
already say what it means that two motions are in
\emph{relative rest}: their world-lines are parallel (i.e.~have
common directional vectors).

We decide that there is no need to differentiate between an
inertial motion and an often used term of \emph{inertial
observer}; the difference, if any, is a~rather psychological
one.

Finally, by an \emph{inertial frame} we mean the class of all
inertial observers remaining in relative rest to each other. We
do not see the need to make this notion more specific, as is
often assumed, by demanding that a~particular basis has been
chosen with the timelike vector along the world line of the
motions in this family.

\section{Metric structure, four-velocities}\label{metric}

We recall two facts from linear algebra:
\begin{itemize}
 \item[$1^\mathrm{o}$] The kernel (zero space) of a~nonzero
     linear form on a~vector space is a~subspace of
     codimension one. Conversely, any such subspace $S$
     determines uniquely up to a~constant factor a~linear
     form $Dt$ such that
    \begin{equation}\label{S}
    x\in S\quad \iff\quad Dt(x)=0\,.
    \end{equation}
 \item[$2^\mathrm{o}$] A real vector quadric $V$ (if
     different from a~subspace) determines uniquely up to
     a~constant factor a~symmetric metric $g$ such that
     \begin{equation}\label{V}
     x\in V\quad \iff\quad g(x,x)=0\,.
     \end{equation}
\end{itemize}
A proof of the second fact for the case of our cone $V$ is
given for completeness in the Appendix.

\subsection{Galilean spacetime}

In the case of the Galilean spacetime we chose the sign of $Dt$
by demanding that
\begin{equation}\label{sign}
 Dt(x)>0\quad \text{for}\quad x\in C_+\,.
\end{equation}
Then $Dt(\ve{PQ})>0$ if $Q$ lies in the future of $P$. The
remaining positive factor in the definition of $Dt$ is fixed
arbitrarily. For an arbitrarily chosen point $P_0$ we fix
a~real value $t(P_0)$. Then there is a~unique affine form
taking this value at $P_0$ and having $Dt$ as its linear part.
This means that for each pair of points $P,Q$ there is
\begin{equation}\label{gal_time}
 t(Q)=t(P)+Dt(\ve{PQ})\,.
\end{equation}
This form determines the universal time in the Galilean
spacetime. The metric structure of this spacetime is now
completed by choosing a~Euclidean metric $h$ on the subspace
$S$. This metric then determines `spatial' metric relations on
each hyperplane $Q+S$ of constant time. One notes that there
are no relations of this kind between points on different
constant time planes. Note also that the relative scale of the
metric tools $Dt$ and $h$ is arbitrary. See Fig.~6 for graphic
representation of the metric structure of GS.

The world-lines of inertial motions pierce precisely at one
point each of the constant time hyperplanes. For each family of
parallel inertial motions there is a~unique directional vector
$u$ for which $Dt(u)=1$. We shall call such vector a~\emph{unit
timelike, future-pointing vector} or the \emph{four-velocity}
of these world-lines.

Having chosen a~particular family of inertial parallel motions
characterized by the four-velocity $u$ one can split the vector
space into time and space parts by
\begin{equation}\label{ts_gal}
 M=L(u)\oplus S\,,
\end{equation}
where $L(u)$ denotes the one-dimensional subspace spanned by
$u$. Observers in the chosen family decompose each vector $x$
into the time and space parts by
\begin{equation}\label{vts_gal}
 x=Dt(x)u+x_u\,,\quad \text{so}\quad x_u\in S\,.
\end{equation}
Note that while $Dt(x)$ does not depend on $u$, the space part
$x_u$ does depend on this vector, that is to say on the family
of parallel inertial motions. The Euclidean scalar product $h$
can be applied to the space parts of any two vectors $x$ and
$y$ and we shall also write
\begin{equation}\label{eucl}
 h(x_u,y_u)=x_u\circ y_u\,.
\end{equation}

\subsection{Special relativity}\label{metric_SR}

In this case $g$ is fixed up to a~real factor by the cone $V$,
as described above. We choose its sign by the convention that
in the canonical basis of $V$ the metric has the signature
$(+1,-1,-1,-1)$. The remaining positive factor is chosen
arbitrarily. The metric structure of the spacetime is
determined completely by~$g$. The vector $x$ is a~timelike
vector when $g(x,x)>0$, and it is a~causal vector when it is
nonzero and $g(x,x)\geq0$. In addition we say that a~vector is
\emph{spacelike} if $g(x,x)<0$. We shall also use the notation
\begin{equation}\label{scalar}
 g(x,y)=x\cdot y\,,\quad\quad x\cdot x=x^2\,.
\end{equation}
See Fig.~7 for the metric properties of vector types.

If $Q$ lies in the future of $P$ then there is a~unique
inertial motion joining them. The \emph{proper time interval}
covered by this motion from $P$ to $Q$ is determined by
\begin{equation}\label{timeint}
 \Delta\tau(P,Q)=\big[g(\ve{PQ},\ve{PQ})\big]^{1/2}\,.
\end{equation}
Let $u=\lambda\ve{PQ}$ with $\lambda>0$ so that $u\in C_+$. If
we demand that $g(u,u)=1$ then $u$ is fixed uniquely by these
conditions and $\lambda=\big[g(\ve{PQ},\ve{PQ})\big]^{-1/2}$.
We call such $u$ a~\emph{unit timelike, future-pointing vector}
or a~\emph{four-velocity}.

A four-velocity $u$ may be used to define a~time variable
correlated with the inertial frame defined by $u$. As in the
Galilean case we fix $t_u(P_0)$ and then there is a~unique
affine form $t_u$ taking this value at $P_0$ and having the
linear form
\begin{equation}\label{tu}
 Dt_u(x)=u\cdot x
\end{equation}
as its linear part. This means that for each pair of points
$P,Q$ there is
\begin{equation}\label{sr_time}
 t_u(Q)=t_u(P)+Dt_u(\ve{PQ})\,.
\end{equation}
Note that if $P$ and $Q$ lie on one $u$-world-line, $Q$ in the
future of $P$, then
\begin{equation}
 Dt_u(\ve{PQ})=\Delta\tau(P,Q)
\end{equation}
so the definition of $Dt_u$ is an extension of the proper time
interval on a $u$-world-line, Eq.~(\ref{timeint}).

Let us denote by $S_u$ the kernel of the form $Dt_u$, which is
the subspace of vectors orthogonal to $u$ with respect to the
metric $g$. Then the hyperplanes $P+S_u$ are the sheets of
constant $t_u$ time. The metric $g$ when restricted to $S_u$
reduces to $-h_u$, where $h_u$ is a~Euclidean metric. Thus the
objects $Dt_u$, $t_u$, $S_u$ and $h_u$ play a~similar role as
$Dt$, $t$, $S$ and $h$ in the Galilean case, but with several
important differences:
\begin{itemize}
 \item[$1^\mathrm{o}$] Here these quantities are not
     universal as in the Galilean case, they are functions
     of the vector $u$; thus they depend on the choice of
     a~family of inertial observers in relative rest.
 \item[$2^\mathrm{o}$] This relative character implies
     weaker status of these quantities as compared to the
     Galilean case.
 \item[$3^\mathrm{o}$] On the other hand the form $Dt_u$
     and the metric $h_u$ are uniquely determined by $g$,
     so their relative scale is unambiguous. This is to be
     contrasted with the Galilean case, where the scale of
     $Dt$ and $h$ could be fixed independently.
\end{itemize}

The decomposition of the vector space $M$ into time and space
parts takes now the form
\begin{equation}\label{ts_sr}
 M=L(u)\oplus S_u\,,\quad
 x=Dt_u(x)u+x_u\,,\quad x_u\in S_u\,,
\end{equation}
see Fig.~8. Note that in this case both $Dt_u(x)$ and $x_u$
depend on $u$, and for different choices of this four-velocity
the space parts $x_u$ lie in different subspaces. For
$x_u,y_u\in S_u$ we shall write $x_u\circ y_u=-x_u\cdot y_u$
and also denote $|x_u|=\sqrt{x_u\circ x_u}$. Then
\begin{equation}\label{ts_metr_sr}
 x\cdot y=Dt_u(x)Dt_u(y)-x_u\circ y_u\,,\quad x^2=(u\cdot x)^2-|x_u|^2\,.
\end{equation}
The scalar product, in contrast to the Galilean case, is
applicable to any vectors. See Fig.~8 and 9 for a~graphic
representation of decompositions and four-velocities, and
Fig.~10 for the dependence of $S_u$ on $u$.

\section{Equivalence of observers, light signals and their
speed}\label{equivalence}

The principle of relativity, i.e.~of the equivalence of
observers, can be now put in the following form:
\begin{itemize}
 \item[$1^\mathrm{o}$] Physical theories do not depend on
     the choice of the inertial frame, i.e.~of the
     four-velocity $u$ determining all inertial motions in
     a~given family.
 \item[$2^\mathrm{o}$] The set of physical states
     conforming with physical theories does not distinguish
     any of the inertial frames.
\end{itemize}
In particular:
\begin{itemize}
 \item[$1^\mathrm{o}$] In SR the Maxwell equations imply
     that the light signals propagate along straight lines
     whose directional vectors lie on $V$, i.e.~$l$ is such
     a~vector iff $g(l,l)=0$. These vectors are called
     therefore \emph{lightlike vectors} and $V$ is called
     the \emph{light-cone}. The Maxwell equations do not
     conform to the principle of relativity in the GS case.
     In this case the only way to avoid clash with the
     principle of relativity is to assume that light
     propagates with infinite speed, i.e.~the directional
     vectors of light rays lie in $S$.
 \item[$2^\mathrm{o}$] If one defines physical units of
     time and space in each inertial frame with the use of
     analogous physical phenomena~then the proportion of
     these units to the geometrical units defined by $Dt$
     and $h$ in the case of Galilean spacetime, and $g$ in
     the case of SR, is the same for all observers.
 \item[$3^\mathrm{o}$] In the SR case if $l$ is lightlike
     and $u$ is any four-velocity, then
     \mbox{$|Dt_u(l)|=|l_u|$} -- light covers in each
     inertial frame a~unit distance in a~unit time in
     geometrical units. If one determines physical units as
     in the preceding point their ratio gives the speed of
     light in all inertial frames in those units.
\end{itemize}

Note that the geometrical objects of the spacetime include
beside metrical tools also the choice of one of the four
orientations (as defined above). The principle of relativity in
the above form does not require the independence of physics of
this choice. As is well-known there are exceptions not
conforming to this extended demand.

\section{Relative velocities and their
composition}\label{velocities}

To be precise the term `four-velocity', although deeply rooted
in the language usually used in SR, is somewhat misleading. In
fact the vector $u$ of an inertial frame simply points in the
direction in which time flows but there is no space translation
for all observers in this frame. To introduce a~more justified
notion of velocity one needs a~reference observer which
`rests'. But `all observers are equal', so one has to say with
respect to which of them one makes the measurement.

Thus we assume there are given two four-velocities $u$ and $u'$
and we want to determine a~velocity of the motion defined by
$u'$ with respect to that defined by~$u$. We propose three
candidates:
\begin{itemize}
 \item[$1^\mathrm{o}$] $\Delta(u',u)=u'-u$,
 \item[$2^\mathrm{o}$] $v_\pr(u',u)=u'_u$,
 \item[$3^\mathrm{o}$] $v(u',u)=u'_u/Dt_u(u')$.
\end{itemize}
The r.h.s.\ in $2^\mathrm{o}$ is formed as in (\ref{vts_gal})
and (\ref{ts_sr}) and the subscript `pr' stands for `proper'.
In~$3^\mathrm{o}$ $Dt_u$ is independent of $u$ in the Galilean
case.

The first of these definitions satisfies the antisymmetry and
chain properties:
\begin{equation}\label{delta}
 \Delta(u',u)=-\Delta(u,u')\,,\quad
 \Delta(u'',u)=\Delta(u'',u')+\Delta(u',u)\,,
\end{equation}
which has obvious interpretational advantages.

\subsection{Galilean spacetime}

In this case $Dt(u')=1$ and $u'_u=u'-Dt(u')u=u'-u$, so all
three definitions coincide and we shall use notation $v(u',u)$
for this quantity (see Fig.~11). We have $v(u',u)\in S$ and
point $3^\mathrm{o}$ above tells us that this vector gives the
change of position of an observer with four-velocity $u'$ with
respect to one with four velocity $u$, undergone in unit time.
The composition of velocities obeys simple vector addition law
(\ref{delta}) (see Fig.~12).

\subsection{Special relativity}

In this case all three definitions are different (see Fig.~13).
The first one has the advantage of the vector addition
composition law (\ref{delta}) (see Fig.~14), but $\Delta(u',u)$
does not lie in any of the subspaces $S_u$ or $S_{u'}$. Rather,
it is in the subspace $S_w$ of the observer with four-velocity
`half way' between $u$ and $u'$:
$w=(u+u')/\sqrt{(u+u')\cdot(u+u')}$.

The second and the third definitions give parallel vectors in
$S_u$. The proper velocity $v_\pr(u',u)$ is the displacement of
the motion along any world-line $P+L(u')$, as seen in the
$u$-frame, undergone during unit time interval as measured on
the world-line (proper time) (see Fig.~15). The velocity
$v(u',u)$ is a~similar displacement but scaled to unit time in
$u$-frame. It is only this latter quantity which is bounded by
$1$ (light velocity as defined in Section~\ref{equivalence}).

The explicit form of the two latter velocities is easily
obtained:
\begin{eqnarray}
 &&v_\pr(u',u)=u'-u'\cdot u\,u\,, \label{vpr}\\
 &&v(u',u)=\frac{u'}{u'\cdot u}-u\,.\label{v}
\end{eqnarray}
Neither of these velocities satisfies the antisymmetry or the
chain rule properties (\ref{delta}). If we write the first of
these equations in the form $u'=u'\cdot u\,u+v_\pr$ and take
the scalar square of both sides we find
\begin{equation}
 (u'\cdot u)^2-|v_\pr|^2=1\,
\end{equation}
(from now on we write $v_\pr\equiv v_\pr(u',u)$, $v\equiv
v(u',u)$). This tells us that the quantities $u'\cdot u$ and
$|v_\pr|$ may be represented as the hyperbolic cosine and
hyperbolic sine of some unique parameter $\psi\geq0$. If we
denote $k=\exp\psi\geq1$ we get the representation
\begin{equation}\label{ukv}
 u'\cdot u=\tfrac{1}{2}(k+k^{-1})\equiv c(k)\,,\quad
 |v_\pr|=\tfrac{1}{2}(k-k^{-1})\equiv s(k)\,,\quad
 |v|=\frac{s(k)}{c(k)}\,.
\end{equation}
Some other useful relations which follow are
\begin{equation}\label{cksk}
 c(k)=\sqrt{1+|v_\pr|^2}=\frac{1}{\sqrt{1-|v|^2}}\,,\quad\quad
 s(k)=\frac{|v|}{\sqrt{1-|v|^2}}\,,
\end{equation}
\begin{equation}\label{k}
 k=|v_\pr|+\sqrt{1+|v_\pr|^2}
 =\lp\frac{1+|v|}{1-|v|}\rp^{1/2}\,.
\end{equation}
We shall find the direct physical interpretation of $k$ in the
next section.

The magnitude of $k$ is invariant with respect to the
interchange of $u$ and $u'$, so if we denote $v'_\pr\equiv
v_\pr(u,u')$ and $v'\equiv v(u,u')$ then we have
\begin{equation}
 |v'_\pr|=|v_\pr|\,,\quad \quad |v'|=|v|\,.
\end{equation}

The motion of an observer with respect to the $u$-frame is
often defined rather in terms of $v_\pr$ or $v$ than $u'$, or
similarly with the role of observers interchanged, and then
\begin{equation}\label{nn}
 \begin{aligned}
 u'&=c(k)u+v_\pr=c(k)(u+v)=c(k)u+s(k)n\,,\\
 u&=c(k)u'+v'_\pr=c(k)(u'+v')=c(k)u'+s(k)n'\,,
 \end{aligned}
\end{equation}
where by $n$ and $n'$ we have denoted the unit spacelike
vectors pointing in the direction of $v$ and $v'$ respectively.
Although the use of $v_\pr$ or $v$ instead of $u'$ may seem
better suited for the point of view of the $u$-frame, one has
to be careful not to project Galilean properties of velocities
to SR. For instance, we have $v'\neq -v$, in contrast to GS.

The composition of velocities of these types is rather
complicated and not very illuminating. The special case of
four-velocities $u$, $u'$, $u''$ lying in one two-dimensional
subspace will be discussed in the next section.

\section{Time measurement}\label{time}

The problem one wants to address here is the following. Two
events $P$ and $Q$ on a~world line with four velocity $u'$ are
separated by the vector $\Delta t' u'$, so the time interval
between them as measured directly by the inertial observer on
this world-line is $\Delta t'$. What time-span $\Delta t$ will
be measured between these events in the frame defined by the
four-velocity $u$?

\subsection{Galilean spacetime}

Here the answer is simple. The spacetime is equipped with the
universal time interval form $Dt$, so there is no doubt how to
measure this interval in any frame. One has
\begin{equation}
 \Delta t=Dt(\Delta t'u')=\Delta t'\,.
\end{equation}

\subsection{Special relativity}

If one employs the frame-dependent time interval form $Dt_u$
described in Section~\ref{metric_SR}, one finds
\begin{equation}\label{dilation}
 \Delta t=Dt_{u}(\Delta t'u')=u\cdot u'\ \Delta t'
 =c(k)\Delta t'\,,
\end{equation}
(notation as in the preceding section). This gives the famous
`time dilation' effect. However, one should be careful to
interpret this result properly. No inertial observer from the
$u$-frame can pass directly both events $P$ and $Q$, thus the
measurement in this frame is by necessity indirect. Observers
on the world-lines $P+L(u)$ and $Q+L(u)$ to establish one
frame-dependent time variable $t_u$ need only to agree on a
choice of a constant time hypersurface to synchronize their
clocks (as the time-interval form $Dt_u$ is known directly to
both of them). After this has been settled (see below) the time
$t_u(P)$ is measured directly by the first observer, and the
time $t_u(Q)$ is measured directly by the other. The difference
$t_u(Q)-t_u(P)$ gives $\Delta t$. See Fig.~16.

The synchronization of clocks can be done by the radar method.
The first observer sends at his time $t_1$ a light signal
towards the other one and receives it back reflected at $t_2$.
Denote by $X$ the event on the world-line of the first observer
at his time $(t_1+t_2)/2$, and by $Y$ the event on the
world-line of the second observer at which the reflection of
the light ray takes place, see Fig.~17. If $l_1$ and $l_2$ are
lightlike vectors as depicted in the figure, then
$(t_2-t_1)u=l_1+l_2$, $\ve{XY}=(l_1-l_2)/2$, so
$u\cdot\ve{XY}=0$. Thus $X$ and $Y$ lie in one hyperplane of
$u$-simultaneity and if the second observer agrees to set his
clock for $(t_1+t_2)/2$ at $Y$, the clocks will be
synchronized.

In real life the time dilation measurement is rarely, if at
all, done this way. Probably the most famous instance of the
dilation effect is the decay of muons produced by cosmic
radiation coming to Earth. Muons are unstable particles with a
characteristic lifetime (in their rest-frames). They are
produced with known energy (so also known velocity) by
scattered cosmic rays. One finds that their mean lifetime in
the Earth-frame is much longer than the characteristic one.
However, what is directly measured is not any time at all! One
measures the distance they cover during their life; then
knowing their relative velocity in the Earth-frame one
\emph{calculates} their lifetime in this frame.

Another type of time measurement is by registering the time of
arrival of light signals. Suppose that two inertial observers
travel along world-lines $P+L(u')$ and $P+L(u)$ respectively
(thus we assume for simplicity that they meet at $P$). Let both
of them set their clocks so as to show $0$ at $P$. The
$u'$-observer sends a light signal towards the $u$-observer at
his time $t'$, which arrives at the $u$-observer's world-line
at the time $t_+$ on that line. Thus one has the equation
$t'u'+l=t_+u$, where $l$ is the lightlike, future-pointing
vector connecting these two events (see Fig.~18). We write this
as
\begin{equation}
 l=t_+u-t'u'\,,\quad\quad  l\cdot l=0\,,\quad\quad l\cdot u>0\,.
\end{equation}
Solving the second equation for $t_+$ one obtains two values
out of which the third condition selects only one:
\begin{equation}
 t_+=u\cdot u'\,t'+\sqrt{(u\cdot u')^2-1}\,|t'|=c(k)t'+s(k)|t'|\,.
\end{equation}
Note that $t',t_+<0$ for observers approaching each other
(parts of world-lines causally preceding $P$) and $t',t_+>0$
for observers moving away from each other (parts of world-lines
causally following $P$). Let now the $u'$-observer send two
signals at times $t'_1$ and $t'_2>t'_1$, either both negative
or both positive, and denote $\Delta t'=t'_2-t'_1$, $\Delta
t_+=t_{+2}-t_{+1}$. Then one finds from the above relation that
\begin{equation}\label{arrival}
 \begin{aligned}
  \Delta t_+&=k^{-1}\Delta t'\quad&&\text{observers moving towards each other}\,,\\
  \Delta t_+&=k\Delta t'\quad&&\text{observers moving away from each other}\,.
 \end{aligned}
\end{equation}
Note that the result is completely different from the `dilation
effect'.

The above connections have a directly observable physical
consequence. The light is a wave phenomenon; the change of its
phase from one ray to another is the same for each of the above
observers. But the times corresponding to the given change of
phase, say $2\pi$, are related as above. Thus the frequencies
of light $\nu'$ and $\nu$ for the two observers are related by
\begin{equation}
 \begin{aligned}
  \nu&=k\,\nu'\quad&&\text{observers moving towards each other}\,,\\
  \nu&=k^{-1}\nu'\quad&&\text{observers moving away from each other}\,.
 \end{aligned}
\end{equation}

With the interpretation of $k$-coefficient given by the second
equation in (\ref{arrival}) we can now find a simple formula
for the composition of velocities (or rather their lengths) in
the special case of three co-planar four-velocities $u$, $u'$,
$u''$. Let the $k$-coefficients be denoted as in Fig.~19. This
figure then also shows that $K=kk'$. Using the last equation in
(\ref{ukv}) and Eq.~(\ref{k}) one finds
\begin{equation}
 |v(u'',u)|=\frac{|v(u',u)|+|v(u'',u')|}{1+|v(u',u)||v(u'',u')|}.
\end{equation}

We end this section with a warning against a popular error in
graphical representations of the time dilation found in many
introductory texts on SR. One of many variants is this: an
individual A is speeding in a rocket towards (or away from)
another individual B, who is busy with some activity. Each of
the individuals is equipped with a clock and A watches (by
`looking') B's activity. The claim then is that A will measure
B's activity to last longer then it lasts for B in agreement
with the time dilation formula. This, however, is wrong; in
fact A receives light signals from B, so his measurement will
give a result obeying one of the cases in Eqs.~(\ref{arrival}).
In fact, for approaching observers, the time in question is
shorter.

\section{Space measurement}

Here we pose the following question. Two parallel world-lines
with four-velocity $u'$ are separated by a~vector $z'$ which is
a `pure space' vector in the $u'$-frame. What is the `pure
space' vector $z$ which separates them in the frame defined by
$u$? These two vectors may be thought of as connecting two
particles in a rigid body in these two frames. This latter
notion has limitations in SR: it runs into difficulty when
accelerations are involved, and then needs an input of dynamics
to be modified. However, as long as only inertial motions are
involved, a rigid body may be \emph{identified} with some
family of parallel world-lines. This body rests in the frame
defined by these world-lines.

\subsection{Galilean spacetime}

Here again the answer is simple: the `pure space' directions
are universally determined by $S$, so
\begin{equation}
 z=z'\in S\,.
\end{equation}

\subsection{Special relativity}

In this case `pure space' means that $u'\cdot z'=u\cdot z=0$.
The condition for $z$ to connect the same two world-lines is
$z=z'+\lambda u'$ with some real~$\lambda$. Taking the scalar
product of this equation with $u$ we find this coefficient and
obtain
\begin{equation}\label{zzprim}
 z=z'-\frac{z'\cdot u}{u'\cdot u}\, u'\,.
\end{equation}
These two vectors can be decomposed as
\begin{equation}\label{zn}
 z'=z'_\perp+\alpha'n'\,,\quad
 z=z_\perp+\alpha n\,,
\end{equation}
where $z'_\perp$ is orthogonal to $u'$ and $n'$ (as defined at
the end of Section~\ref{time}), $z_\perp$ is orthogonal to $u$
and $n$, and $\alpha$, $\alpha'$ are numerical constants. Note
that $z'_\perp$ and $z_\perp$ are equivalently identified as
parts of $z'$ and $z$ orthogonal both to $u$ and $u'$. Taking
the scalar product of Eq.~(\ref{zzprim}) with $u'$ we find
$z\cdot u'=-z'\cdot u/u'\cdot u$. Using now Eqs.~(\ref{nn}) and
(\ref{zn})  we find after some simple algebra
\begin{equation}\label{zperp}
 z_\perp=z'_\perp\,,\quad \quad \alpha=-\frac{\alpha'}{c(k)}\,.
\end{equation}
The second of these equations describes the effect of the so
called `length contraction', whose popular formulation could
run as: `the dimensions parallel to the relative velocity
measured by the moving observer are by the factor $1/c(k)$
shorter then those measured by the observer in rest with
respect to the object being measured'.  However, one should
note that this formulation and the term `contraction' are
somewhat misleading:
\begin{itemize}
 \item[$1^\mathrm{o}$] The vectors $z'$ and $z$ connect two
     different pairs of events on the two world-lines
     considered, nothing is being `contracted'. Events
     separated by $z'$ are simultaneous in the rest frame
     of the `rigid body', while those separated by $z$ are
     simultaneous for the moving observer.
 \item[$2^\mathrm{o}$] The vectors $n'$ and $n$ (pointing
     in the directions of the two respective velocities)
     are not even parallel, so for each of the frames the
     term `parallel to the velocity' means something
     different.
\end{itemize}
Figure 20 illustrates the situation for the special case
$z'_\perp=z_\perp=0$, which means that for the $u$-observer the
rigid rod with ends on the two world-lines moves parallelly to
its axis.

The proper understanding of the above dismisses various `length
contraction paradoxes' in SR.\cite{con} The key to all of them
is a cautious analysis of the relation between various vectors
involved in the problem.

We illustrate this with a geometrical situation whose variants
lie at the base of most of these effects. Suppose we have two
pairs of parallel world-lines: $P+L(u')$, $Q+L(u')$, and
$P+L(u)$, $Q+L(u)$, so that the first lines in these pairs
intersect at $P$, and the second lines intersect at $Q$.
Physically this may be thought of as modeling two rigid rods in
relative motion, the ends of the first and the second rod
described by the lines in the first and in the second pair
respectively. The `front' ends of the rods meet at some point
and similarly the `back' ends meet at some other point. Let
$z'$ and $w$ be the `pure space' vectors (in respective
rest-frames) connecting the ends of rods and denote
$x=\ve{PQ}$. (See Fig.~21. The picture might suggest that the
rods are bound to clash and cannot `go through'. This is
because we lack in the picture the fourth dimension, which may
be used to slightly detach the rods.) Then one has
\begin{equation}\label{x}
 x=z'+\mu'u'=w+\nu u
\end{equation}
with some constants $\mu'$, $\nu$. We decompose $z'$ as in the
first Eq.~(\ref{zn}) and similarly write
\begin{equation}
 w=w_\perp+\beta n\,,\quad
 w'=w_\perp-\frac{\beta}{c(k)}\,n'\,,
\end{equation}
(the second formula obtained in analogy with Eqs.~(\ref{zn})
and (\ref{zperp}) is written down for later use). As $n$ and
$n'$ can be expressed as linear combinations of $u$ and $u'$
(see Eq.~(\ref{nn})), the consistency condition for the second
equation in (\ref{x}) is
\begin{equation}\label{zpw}
 z'_\perp=w_\perp\,,
\end{equation}
and then the constants $\mu'$ and $\nu$ have unique solutions,
which we do not need to write down explicitly.

The geometry of the situation is clear and no interpretational
difficulty arises if one insists on this four-dimensional
picture. However, if one uses the `length contraction' language
`paradoxes' easily arise. Suppose, for instance, that the
vector $x$ is spacelike (as in Fig.~21) and consider any
four-velocity orthogonal to $x$. Then the intersecting of lines
has this interpretation: in each of these frames the two rods
pass each other parallelly, with both respective ends
simultaneously coming into contact. But now the `paradoxical'
problem arises: if we go to some other frame not in this
family, then due to different velocities of the two rods they
will change their size in different way, so the ends cannot
meet. The simple explanation is, of course, that what is
simultaneous in one frame usually is not simultaneous in
another, which falsifies the above conclusion. And even more,
the rods moving parallelly in one frame usually do not remain
parallel in another.

To illustrate the last point suppose that in the above
geometrical setting $x=w$, i.e. the rods are parallel and of
equal length in the $u$-frame. This means that $w=z$, and
decomposing these vectors as before we find
$\alpha'=-c(k)\beta$. Using this and Eq.~(\ref{zpw}) we find
\begin{equation}
 w'=w_\perp-\frac{\beta}{c(k)}\,n'\,,\quad
 z'=w_\perp-c(k)\beta\,n'\,.
\end{equation}
These vectors are parallel if, and only if $w_\perp=0$ or
$\beta=0$. In all other cases rods move in the $u'$-frame askew
to each other. This is illustrated in Fig.~22.

\section{Non-inertial motions, proper time, simultaneity}

Inertial motions, as we have seen, have a special role to play
for the interpretation of the geometry of spacetime. However,
the picture would not be complete without mentioning other,
non-inertial, motions. Straight lines are special examples in
the more general class of curves. A \emph{regular curve} may be
defined as a set of points obtained as values of a
differentiable mapping $\lambda\mapsto P(\lambda)$, where
$\lambda$ is a real parameter taking values in some (finite or
not) interval on the real axis. The curve is invariant under a
change of parameter $\lambda=f(\lambda')$, where $f$ is
differentiable together with its inverse. Each regular curve
has at each its point $P(\lambda)$ a \emph{tangent vector}
defined as $dP(\lambda)/d\lambda$. The extension of tangent
vectors changes with the change of parameter (but the tangent
straight lines they generate remain unchanged).

We now define a general \emph{world-line} as a curve with a
four-velocity as its tangent vector at each its point.  We say
that $\tau$ is a \emph{proper time} of a world-line if it has
the form $\tau\mapsto P(\tau)$ and the equation
\begin{equation}\label{accdiff}
 \frac{dP(\tau)}{d\tau}=u(\tau)
\end{equation}
defines at each point the tangent four-velocity $u(\tau)$.
Physically proper time intervals are measured by clocks
traveling along the world-line. Integrating the above equation
one obtains
\begin{equation}\label{accint}
 \ve{P_1P_2}=\int_{\tau_1}^{\tau_2}u(\tau)\,d\tau\,,\quad\text{where}\quad
 P_i=P(\tau_i)\,.
\end{equation}
Note that sums of four-velocities are future-pointing timelike
vectors, so $P_2$ is in the future of $P_1$. One introduces
also the concept of the \emph{four-acceleration}:
\begin{equation}\label{acc}
 a(\tau)=\frac{du(\tau)}{d\tau}\,.
\end{equation}
Note that acceleration, like relative velocity, points in a
`purely spatial' direction:
\begin{equation}
 \begin{split}
 &\text{GS:}\quad\quad Dt(a(\tau))=\frac{d}{d\tau}Dt(u(\tau))=0\,,\\
 &\text{SR:}\quad\quad Dt_{u(\tau)}(a(\tau))=u(\tau)\cdot a(\tau)
 =\tfrac{1}{2}\frac{d}{d\tau}[u(\tau)]^2=0\,.
 \end{split}
\end{equation}
However, unlike relative velocity, the acceleration is absolute
-- it does not need a reference observer.

We now want to find
\begin{itemize}
 \item[$1^\mathrm{o}$] what is the relation of the proper
     time to affine time functions defined earlier,
 \item[$2^\mathrm{o}$] does the presence of acceleration
     influence the concept of simultaneity?
\end{itemize}

\subsection{Galilean spacetime}

We apply the linear form $Dt$ to both sides of
Eq.~(\ref{accint}) and find
\begin{equation}
 t(P_2)-t(P_1)=Dt(\ve{P_1P_2})
 =\int_{\tau_1}^{\tau_2}Dt(u(\tau))\,d\tau=\tau_2-\tau_1\,.
\end{equation}
Thus the proper time intervals are identical with the absolute
time intervals. Also, the notion of simultaneity is in no way
influenced by accelerations.

\subsection{Special relativity}

Here we take the form $Dt_u$ and then proceed as in the
Galilean case to find
\begin{equation}
 t_u(P_2)-t_u(P_1)
 =\int_{\tau_1}^{\tau_2}u\cdot u(\tau)\,d\tau
 \geq \tau_2-\tau_1\,.
\end{equation}
Therefore the proper time interval is always smaller than any
affine time function interval, except for the case when
$u(\tau)\equiv u$. The latter case gives simply
$P(\tau)=P(\tau_1)+(\tau-\tau_1)u$, which is an inertial
motion; proper time intervals are then equal to the
$u$-inertial time intervals on that line. In general this is
not the case. However, put $\tau_1=\tau$, $\tau_2=\tau+d\tau$
and $u=u(\tau)$. Then we find
\begin{equation}
 t_{u(\tau)}(P(\tau+d\tau))-t_{u(\tau)}(P(\tau))=
 d\tau\,,
\end{equation}
so locally the proper time interval is equal to the time
interval as defined earlier for inertial motions.

With accelerated motions in play it is now possible to let two
general observers start from $P_1$, take different routes, and
then meet again at $P_2$. In general their clocks will show
different time intervals between these two events. In
particular, let the first observer go straight from $P_1$ to
$P_2$ along an inertial world-line, and let $u$ be his
four-velocity. Then his clock will show the interval
$t_u(P_2)-t_u(P_1)$, which is always more than the reading of
the proper time interval for any accelerated observer. There is
no paradox here (the famous `twin paradox') -- the
accelerations, as noted above, are absolute, so there is no
symmetry between the observers.

Consider now simultaneity. Suppose that for an observer on the
world-line $P(\tau)$ we can extend this notion in the way
determined by his local position and four-velocity: event $X$
is from his point of view simultaneous with the event $P(\tau)$
iff $\ve{P(\tau)X}\cdot u(\tau)=0$. However, this leads to
conceptual difficulties. To see this suppose the observer
crosses $P_1$ with four-velocity $u_1$ and then $P_2$ with
four-velocity $u_2$. The two corresponding simultaneity
hyperplanes cross on the 2-plane of events $X$ determined by
the linear system
\begin{equation}
 \ve{P_iX}\cdot u_i=0\,,\quad i=1,2\,.
\end{equation}
Take any event $X$ on this 2-plane and put $X'_i=X+\ve{P_iX}$.
We have $\ve{P_iX'_i}=2\ve{P_iX}$, so $X'_i$ is simultaneous
with $P_i$. At the same time there is
$\ve{X'_1X'_2}=-\ve{P_1P_2}$. Therefore $X'_2$ is in the past
of $X'_1$. Thus an event which according to the above
definition is simultaneous with $P_1$ turns out to be in the
future of an event simultaneous with a later event $P_2$ (see
Fig.~23).

This difficulty should by no means be interpreted as an
argument against the objectivity of the `direction of time
flow'. This latter notion should be simply identified with the
choice of the causal orientation and the emerging partial order
$Q\geq P$, as discussed in Section~\ref{orientation}. The
difficulty rather points to the weakness of the notion of
simultaneity, its restricted applicability and, to some degree,
its conventional character. It also shows that the strict
`dilation' and `contraction' problems are of rather academic
nature.

\section{Four-momentum, four-angular momentum and their
conservation}\label{momentum}

The four-momentum of a~particle with mass $m_1$ and
four-velocity $u_1$ is given by
\begin{equation}\label{mom}
 p_1=m_1u_1\,.
\end{equation}
If one chooses a~reference point $O$ and $x_1$ is a~vector from
this point to the position of the particle then the
four-momentum tensor is defined by
\begin{equation}\label{ang}
 L_1=2x_1\wedge p_1\,.
\end{equation}
Let $p_1,\ldots,p_k$ be the initial and $p'_1,\ldots p'_l$  the
final four-momenta~in a~conservative mechanical process. The
invariant laws of momentum and angular momentum conservation
say
\begin{equation}\label{conservation}
 \sum_{i=1}^kp_i=\sum_{j=1}^lp'_j\,,\quad\quad
 \sum_{i=1}^kL_i=\sum_{j=1}^lL'_j\,.
\end{equation}

\subsection{Galilean spacetime}

Here the mass is an invariant of the four-momentum given by
$m_1=Dt(p_1)$. The decomposition of the four-momentum with
respect to the frame defined by the four-velocity $u$ is thus
\begin{equation}\label{mom_gal}
 p_1=m_1u+p_{1u}\,
\end{equation}
see Fig.~24. We see thus that the law of conservation of mass
and the law of conservation of momentum are aspects of one
observer-invariant law of conservation of four-momentum.

\subsection{Special relativity}

The mass again is an invariant, but formed in another way:
$p_1\cdot p_1=m_1^2$. Then in the $u$-frame we have
\begin{equation}\label{mom_sr}
 p_1=E_{1u}u+p_{1u}\,,\quad E_{1u}^2-|p_{1u}|^2=m_1^2\,,
\end{equation}
see Fig.~25. $E_{1u}$ has the interpretation of the energy as
seen in the chosen frame. Now the aspects of the
observer-invariant law of conservation of four-momentum are
laws of energy and momentum conservation, while the sum of
masses needs not to be conserved.

We observe that geometrical analogy is:
\begin{center}
 Galilean mass\quad $\leftrightarrow$ \quad Einsteinian energy
\end{center}
(and not energy $\leftrightarrow$ energy). This analogy is
further confirmed when one considers the time-space part of the
conservation of four-angular momentum. For freely moving
particles one obtains the law of uniform motion of center of
mass in the Galilean case, and of center of energy in the SR
case.

\section{Galilean kinetic energy}

The question then arises what is the geometrical status of the
Galilean kinetic energy and does its conservation have an
invariant character.

To answer this observe that while there is no geometrical
numerical  invariant formed out of space-part of a~single
timelike vector, one can form a~respective invariant for a~pair
of such vectors. Let $Dt(p_i)=m_i$, $i=1,2$, and let $u$ be any
four-velocity. Then $p_i=m_iu+p_i{}_u$, so that
\begin{equation}\label{inv}
 \frac{p_1}{m_1}-\frac{p_2}{m_2}=
 \frac{p_1{}_u}{m_1}-\frac{p_2{}_u}{m_2}\in S\,.
\end{equation}
Thus the number
\begin{equation}\label{invariant}
 d(p_1,p_2)=\frac{m_1m_2}{2}\,\left|\frac{p_1{}_u}{m_1}-\frac{p_2{}_u}{m_2}\right|^2\geq0
\end{equation}
does not depend on $u$ (see Fig.~26). For momenta
$p_1,\ldots,p_k$ it is now easy to show, that
\begin{equation}\label{energy}
 \sum_{i,j=1}^kd(p_i,p_j)=2ME-|P_u|^2\geq0\,,
\end{equation}
where
\begin{equation}\label{P}
 P=\sum_{i=1}^kp_i\,,\quad P=Mu+P_u\,,\quad
 E=\sum_{i=1}^k\frac{|p_i{}_u|^2}{2m_i}\,.
\end{equation}
We learn two facts:
\begin{itemize}
 \item[$1^\mathrm{o}$] If the total four-momentum is
     conserved, then the condition of energy conservation
     is Galilean invariant.
 \item[$2^\mathrm{o}$] There is always $E\geq |P_u|^2/2M$,
     and the equality holds if, and only if, all momenta
     are parallel.
 \end{itemize}

\section{Celestial sphere}

We fix a~reference point $O$ and consider all light rays coming
into this point. Imagine a~world-line of an inertial observer
with four-velocity $u$ passes through this point. At this point
the observer positions the space directions from which all
light rays arrive. We want to find how the picture obtained in
this way depends on the four-velocity $u$ of the observer.

\subsection{Galilean spacetime}

Here we assume that the light rays propagate with infinite
speed. Thus the straight lines of the rays lie in the
hyperplane $O+S$, and their directional vectors are in $S$. But
for such vectors the decomposition (\ref{vts_gal}) is trivial
and independent of $u$. Therefore the picture formed by light
on the celestial sphere is independent of the choice of
particular observer crossing the point $O$.

\subsection{Special relativity}

A~light ray with the directional past-pointing vector $-l\in V$
comes from the space direction pointed by the unit spacelike
vector
\begin{equation}\label{rlu}
 r(l,u)=\frac{-l_u}{|l_u|}
 = -\frac{l-u\cdot l\,u}{u\cdot l}=-\frac{l}{u\cdot l}+u\,,
\end{equation}
where we have used the fact that $|l_u|^2=-l_u\cdot l_u=(u\cdot
l)^2$ (see Fig.~27). If $u'$ is the four-velocity of another
observer passing $O$ and we denote for brevity $r=r(l,u)$,
$r'=r(l,u')$ then we find
\begin{equation}
 \frac{u'\cdot l}{u\cdot l}=(u-r)\cdot u'=c(k)+s(k)\,n\circ r\,.
\end{equation}
Using this and Eq.~(\ref{rlu}) for $r$ and $r'$ we find the
transformation $r\mapsto r'$ of the celestial sphere of the
$u$-observer to the sphere of the $u'$-observer:
\begin{equation}\label{trcs}
 r'=u'+\frac{r-u}{c(k)+s(k)\,n\circ r}\,.
\end{equation}
Taking the scalar product of this equation with $u$ we find, in
particular, the well-known aberration formula:
\begin{equation}
 n'\circ r'=-\,\frac{s(k)+c(k)\,n\circ r}{c(k)+s(k)\,n\circ r}\,
\end{equation}
(the difference in signs is due to the direction of $n$ and
$n'$).

A small variation of the direction of the light ray induces
small variations $\delta r$ and $\delta r'$, which are tangent
to the two respective celestial spheres. The linear
transformation $\delta r\mapsto \delta r'$ is found by varying
Eq.~(\ref{trcs}):
\begin{equation}
 \delta r'=\frac{\delta r}{c(k)+s(k)\,n\circ r}+
 \frac{n\circ\delta r}{[c(k)+s(k)\,n\circ r]^2}\,(u-r)\,.
\end{equation}
Taking now two different variations $\delta_1$ and $\delta_2$
and using the constraints $u\cdot\delta r=r\cdot\delta r=0$ we
find
\begin{equation}
 \delta_1r'\circ\delta_2r'=
 \frac{\delta_1r\circ\delta_2r}{[c(k)+s(k)\,n\circ r]^2}\,.
\end{equation}
This equation tells us that the linear transformation $\delta
r\mapsto \delta r'$ differs only by the factor
$[c(k)+s(k)\,n\circ r]^{-1}$ from an isometric transformation.
Thus locally (in the first order in $\delta r$) the picture
registered on the celestial sphere scales by this factor
without a change of the shape (the angles).\cite{ter}

Larger areas on the celestial sphere lose this scaling property
and undergo more complicated transformations. However, one
feature of the local transformation survives. To find it chose
a~spacelike vector $z$, $z^2<0$, and consider among vectors
$-l$ all those which satisfy the equation
\begin{equation}\label{circle}
 z\cdot l=0\,.
\end{equation}
Using the geometrical quantities correlated to $u$ the
spacelike character of $z$ is written down as $(u\cdot
z)^2<|z_u|^2$ and the above condition on $l$'s takes the form
\begin{equation}\label{circle_u}
 r(l,u)\circ\frac{z_u}{|z_u|}=-\frac{u\cdot z}{|z_u|}=\cos[\phi(z,u)]\,,
\end{equation}
where the last equality defines the angle $\phi(z,u)$. This
equation tells us that the vectors $r(l,u)$ are all those which
form the angle $\phi(z,u)$ with the vector $z_u/|z_u|$. Thus
they form a~circle on the celestial sphere. This fact is
independent of the choice of a~particular observer (its vector
$u$) crossing the point $O$. However, the angle $\phi(z,u)$
does depend on this choice. Note in particular that if
Eq.~(\ref{circle}) determines a~`great circle' for the observer
with four-velocity $u$ (i.e.~$\phi(z,u)=\pi/2$), this circle
will in general cease to be `great' for the one with the
four-velocity $u'$. The exceptional cases when `great' goes to
`great' are those determined by $z$ orthogonal both to $u$ and
$u'$.

To summarize, the picture obtained on the celestial sphere
undergoes deformation from one observer to another, but in such
a~way that angles are conserved and circles become circles,
although the `greatness' property is usually not conserved.
This is illustrated in Figs.~28 and 29.

\section{Acknowledgements}

I am grateful to my colleague Piotr Bizoń for his suggestion to
expand what originally was a lecture presentation into this
article, and for careful reading of the manuscript.

\section{Appendix}

\begin{theor}
The cone $V$ determines uniquely up to a~constant factor
a~symmetric bilinear form $g$ such that\quad $x\in V\ \iff\
g(x,x)=0$.
\end{theor}
\begin{proof}
In a~canonical basis $V$ takes the form given in
Eq.~(\ref{cone}), which is equivalent to
\mbox{$x^0=\pm\sqrt{\sum_{i=1}^3(x^i)^2}$}. If this implies
$g(x,x)=\sum_{\mu,\nu=0}^3g_{\mu\nu}x^\mu x^\nu=0$, then the
conditions
\begin{equation*}
 \lp g_{ik}+g_{00}\delta_{ik}\rp x^ix^k \pm 2g_{0i}
 {\textstyle\sqrt{\sum_{i=1}^3(x^i)^2}\, x^i}=0
\end{equation*}
must be satisfied identically (for any numbers $x^i$,
$i=1,2,3$). Thus \mbox{$g_{0i}=0$}, $i=1,2,3$, and
$g_{ik}+g_{00}\delta_{ik}=0$, $i,k=1,2,3$. Therefore in this
frame \mbox{$g(x,y)=g_{00}\lp x^0y^0-\sum_{i=1}^3x^iy^i\rp$}.
\end{proof}

 \eject
 \vspace*{-3cm}
\begin{center}
\includegraphics[width=.7\textwidth]{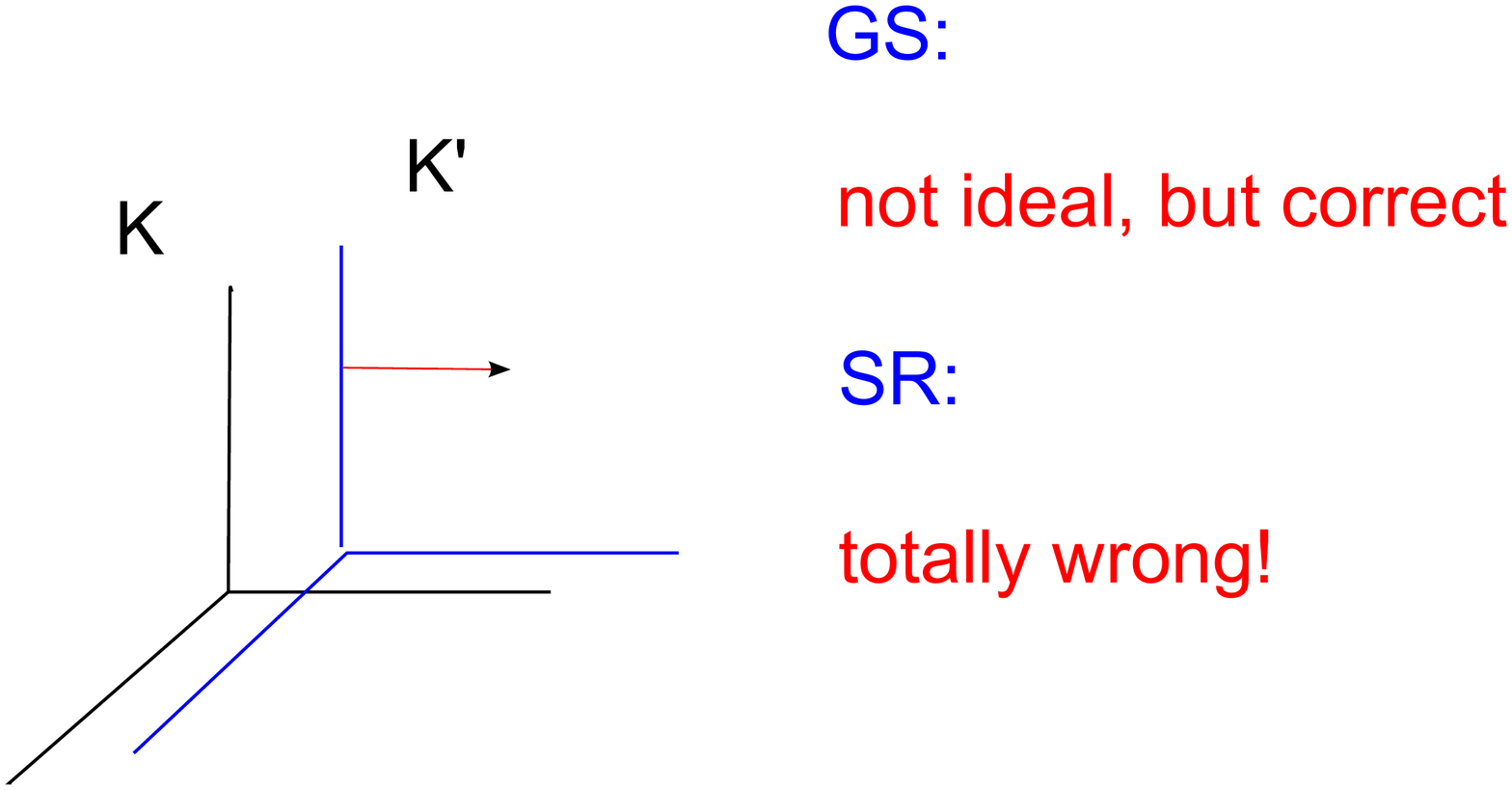}\vspace{-.5cm}\\
Fig.~1. Reference frames -- a popular picture.\\
\includegraphics[width=.7\textwidth]{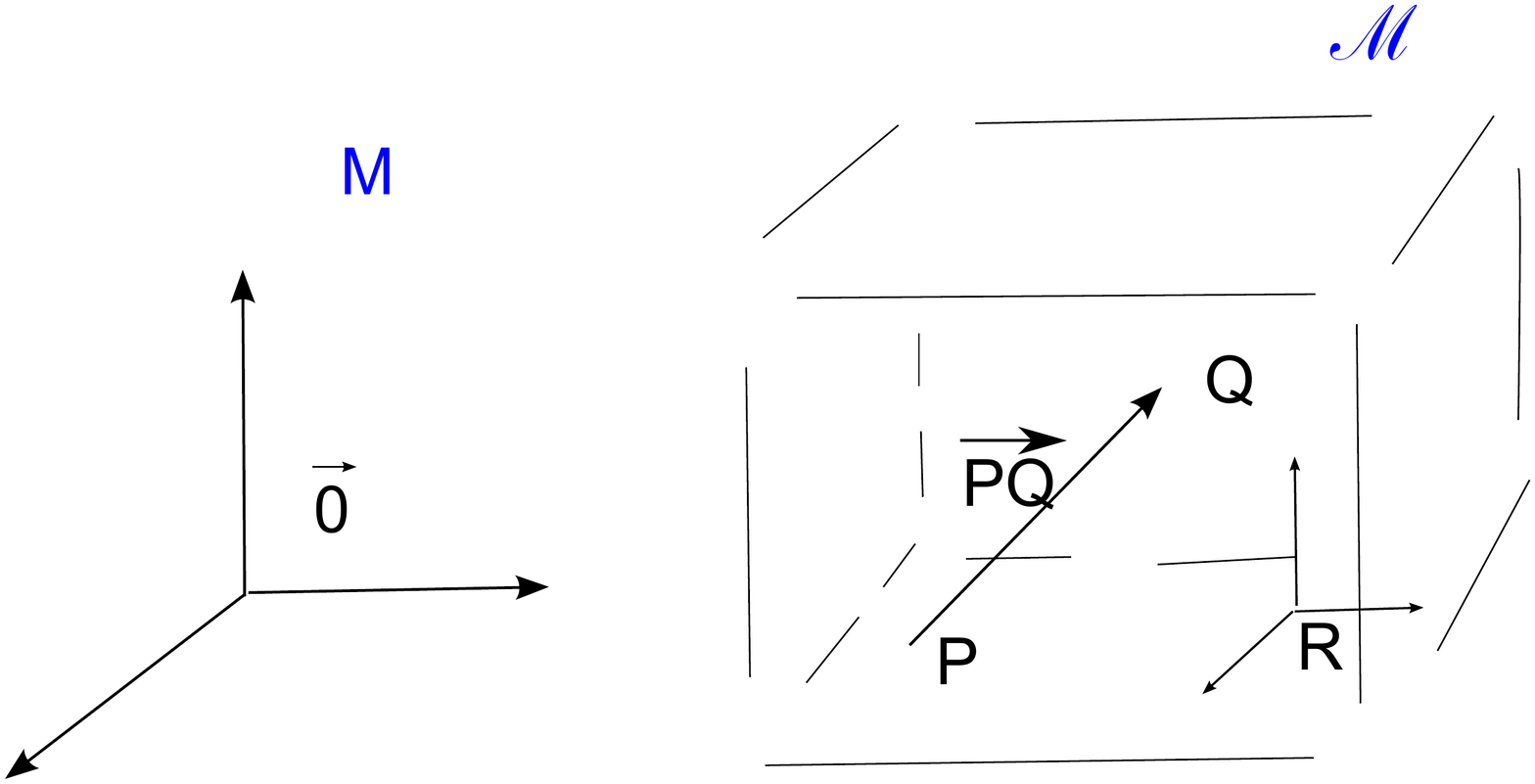}\vspace{-.5cm}\\
 Fig.~2. Vector and affine space.\\
\includegraphics[width=.7\textwidth]{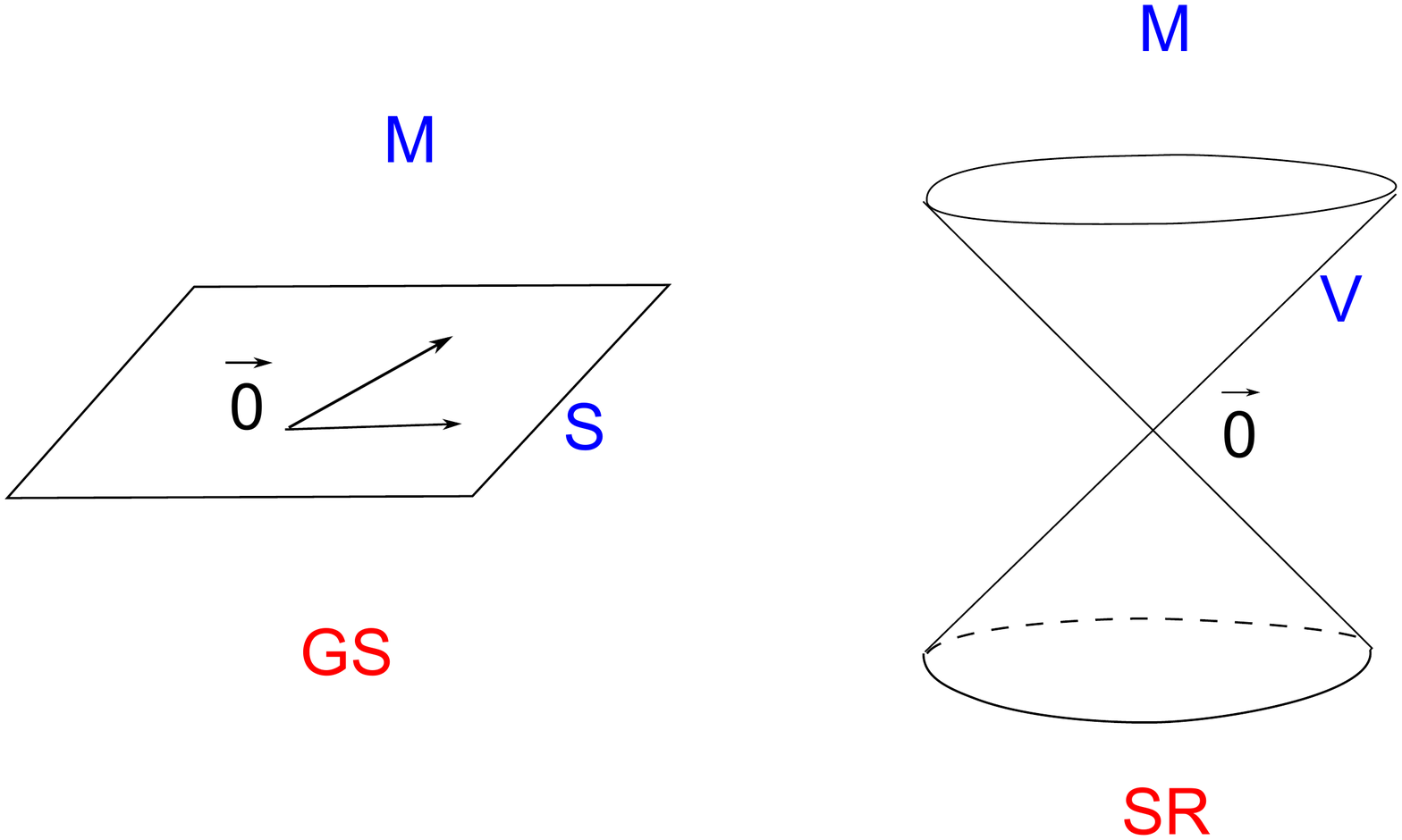}\vspace{-.5cm}\\
Fig.~3. Causal structure.
 \eject

 \vspace*{-3cm}
\includegraphics[width=.7\textwidth]{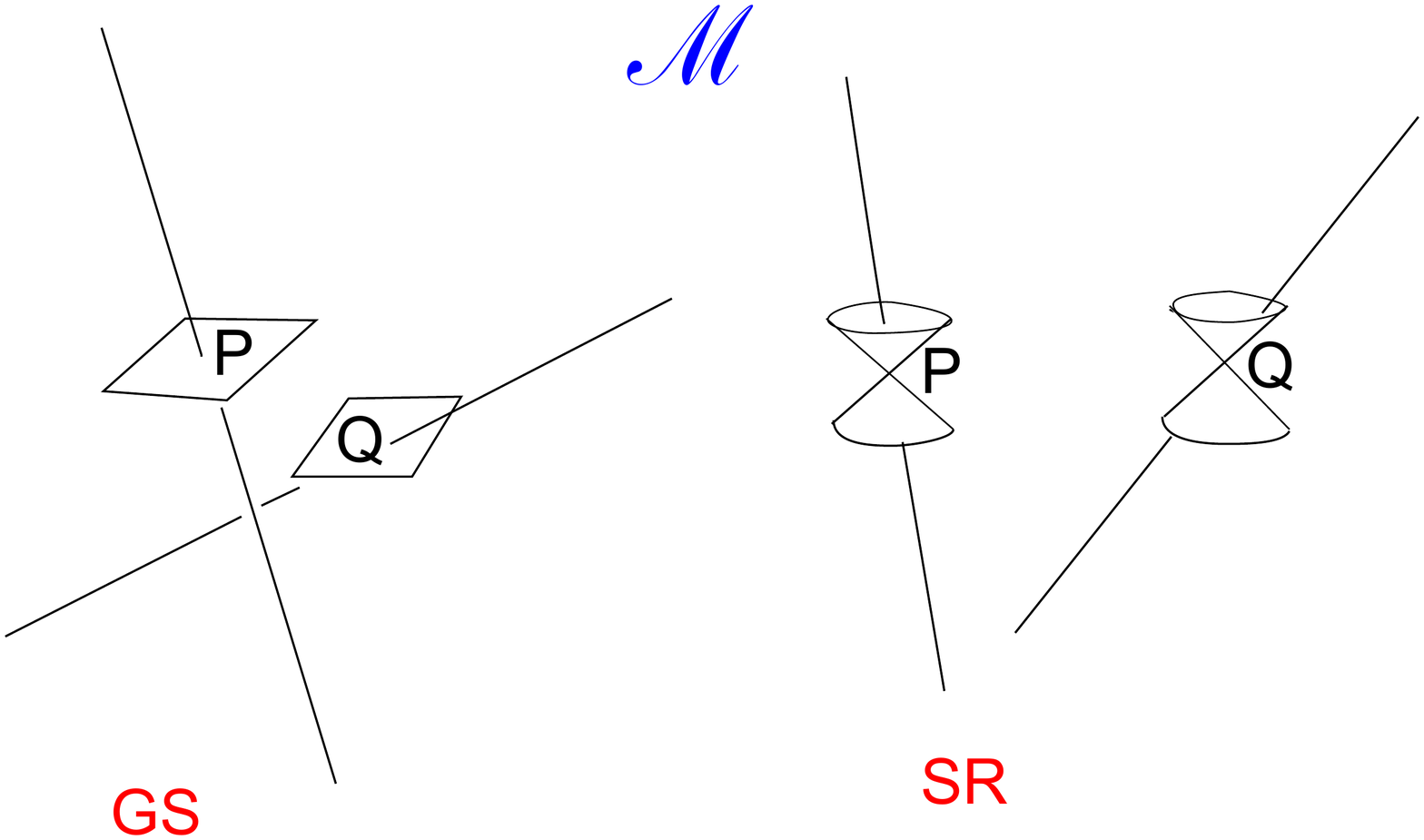}\vspace{-.5cm}\\
Fig.~4. Inertial motions.\\
\includegraphics[width=.7\textwidth]{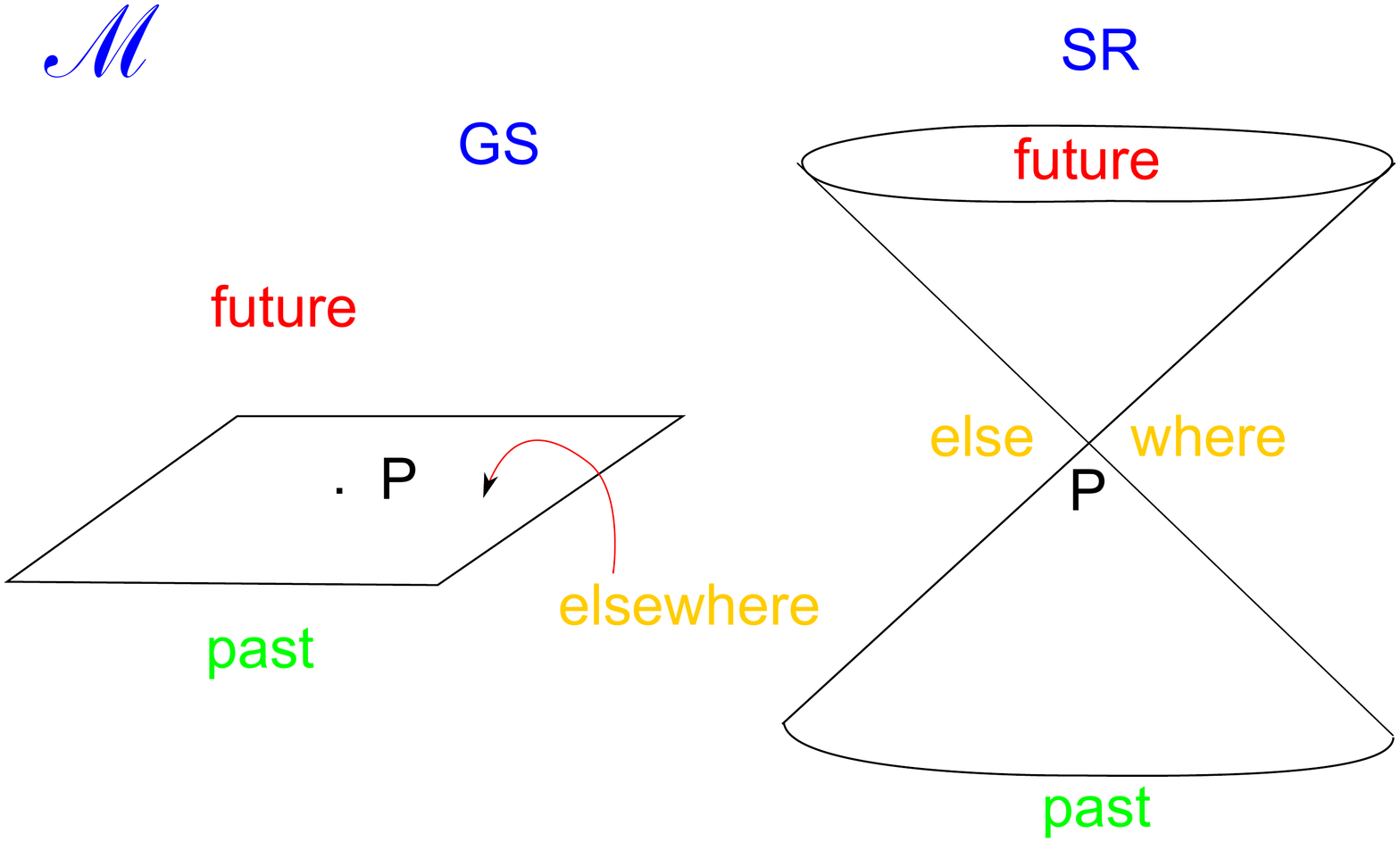}\vspace{-.5cm}\\
Fig.~5. Past, future, elswhere.\\
\includegraphics[width=.7\textwidth]{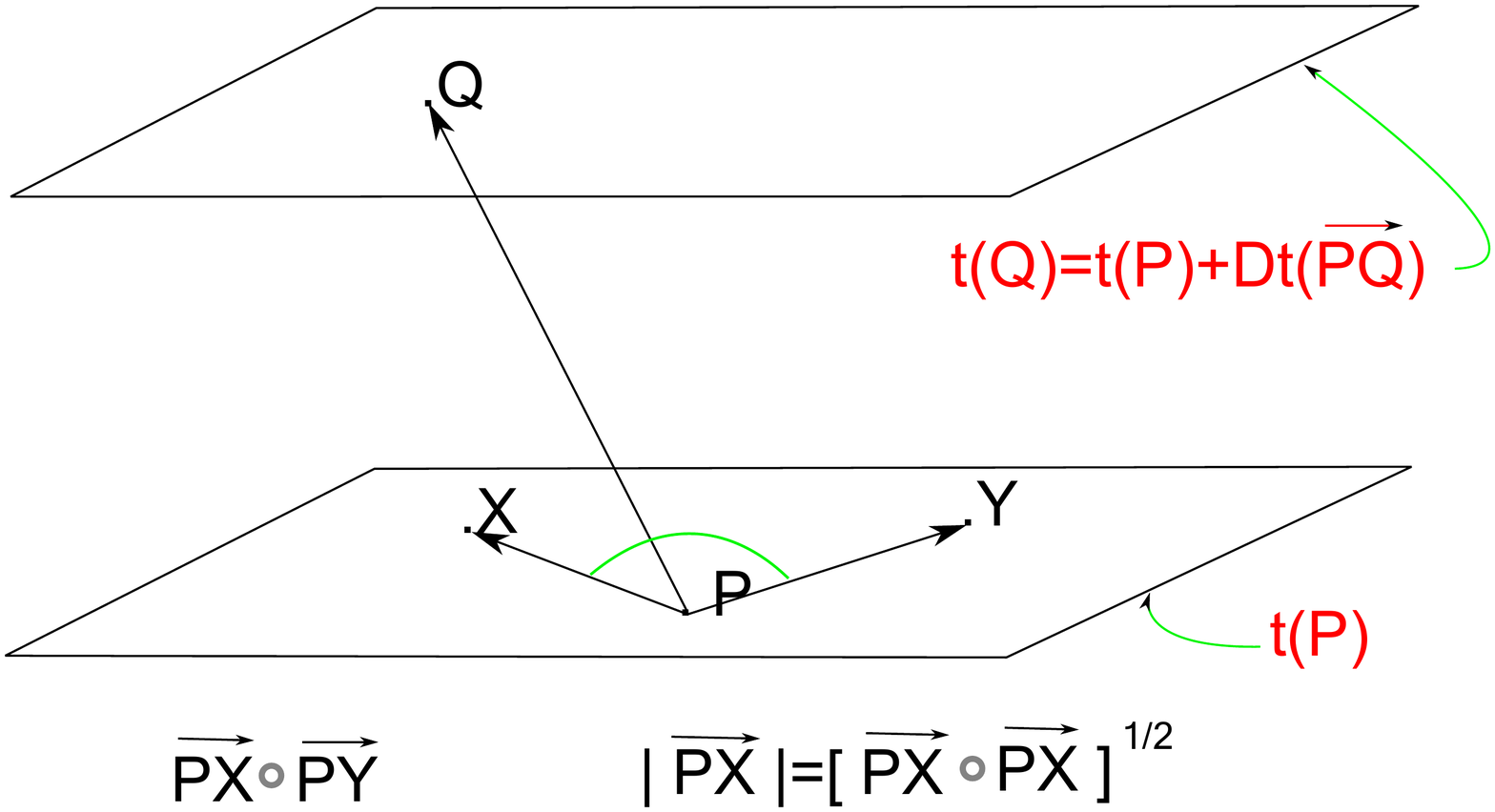}\vspace{-.5cm}\\
Fig.~6. Metric structure of GS.
 \eject

 \vspace*{-3cm}
\includegraphics[width=.7\textwidth]{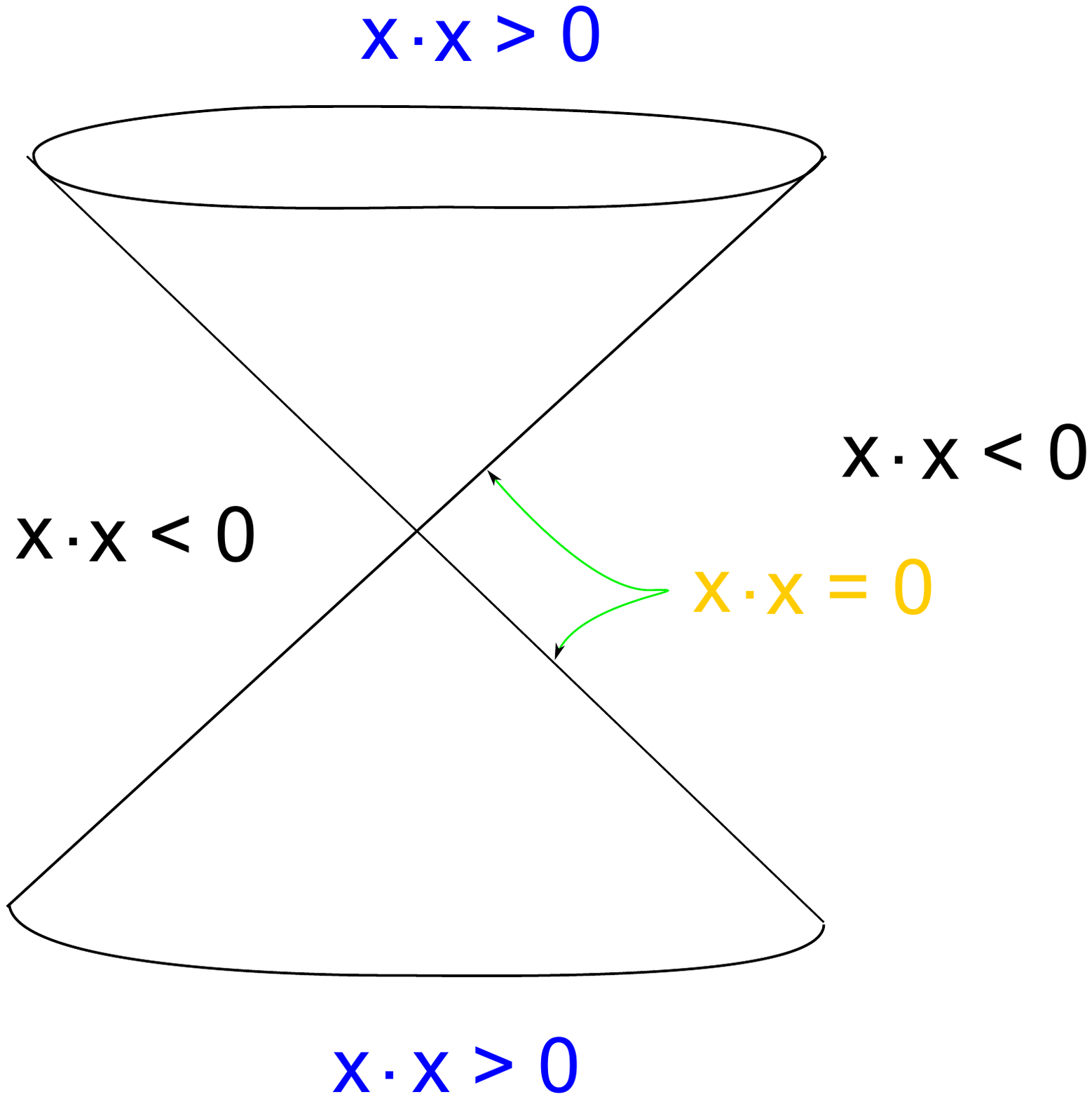}\vspace{-.5cm}\\
Fig.~7. Scalar product in SR.\\
\includegraphics[width=.7\textwidth]{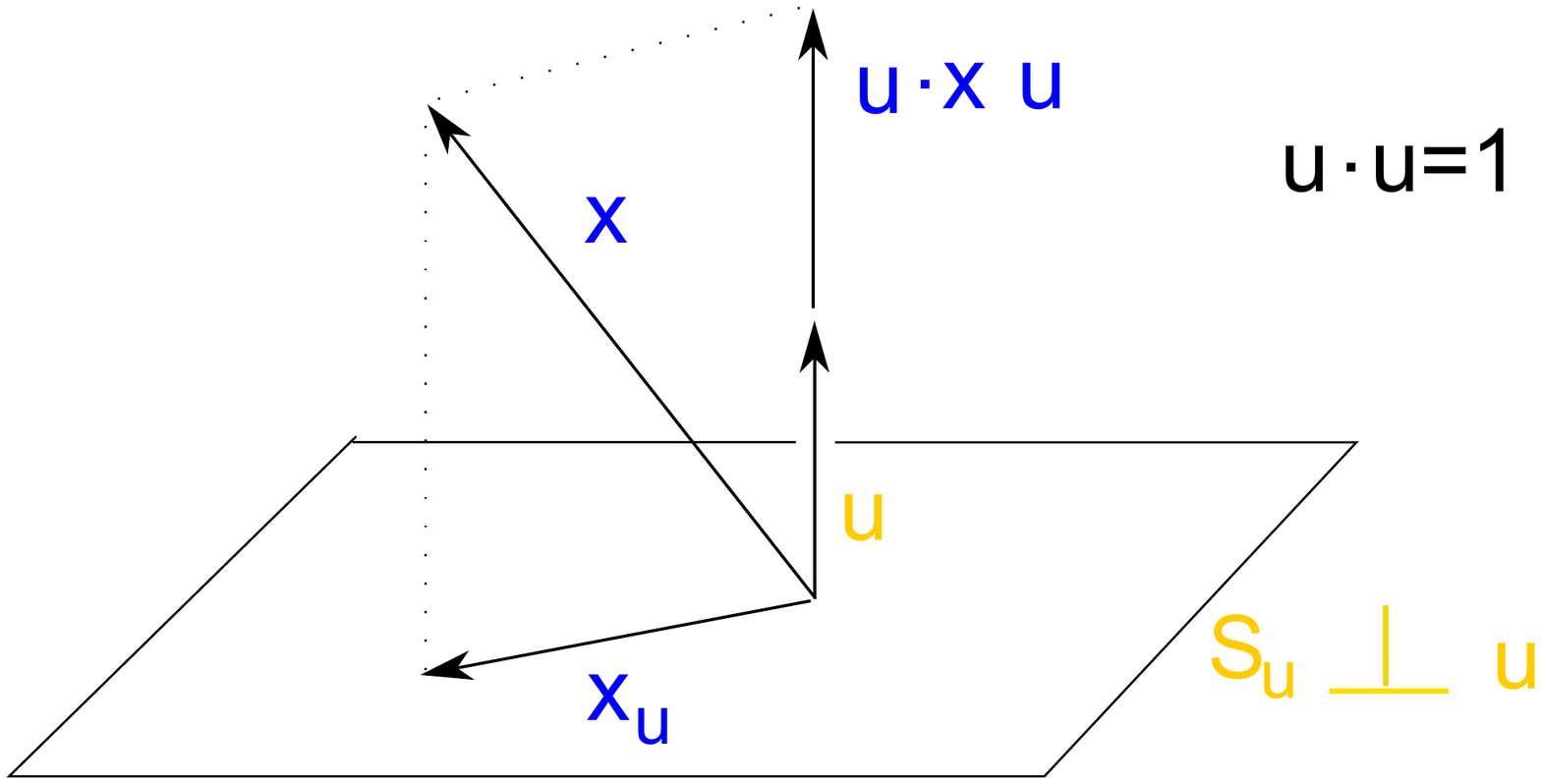}\vspace{-.5cm}\\
Fig.~8. Metric structure of SR.\\
\includegraphics[width=.7\textwidth]{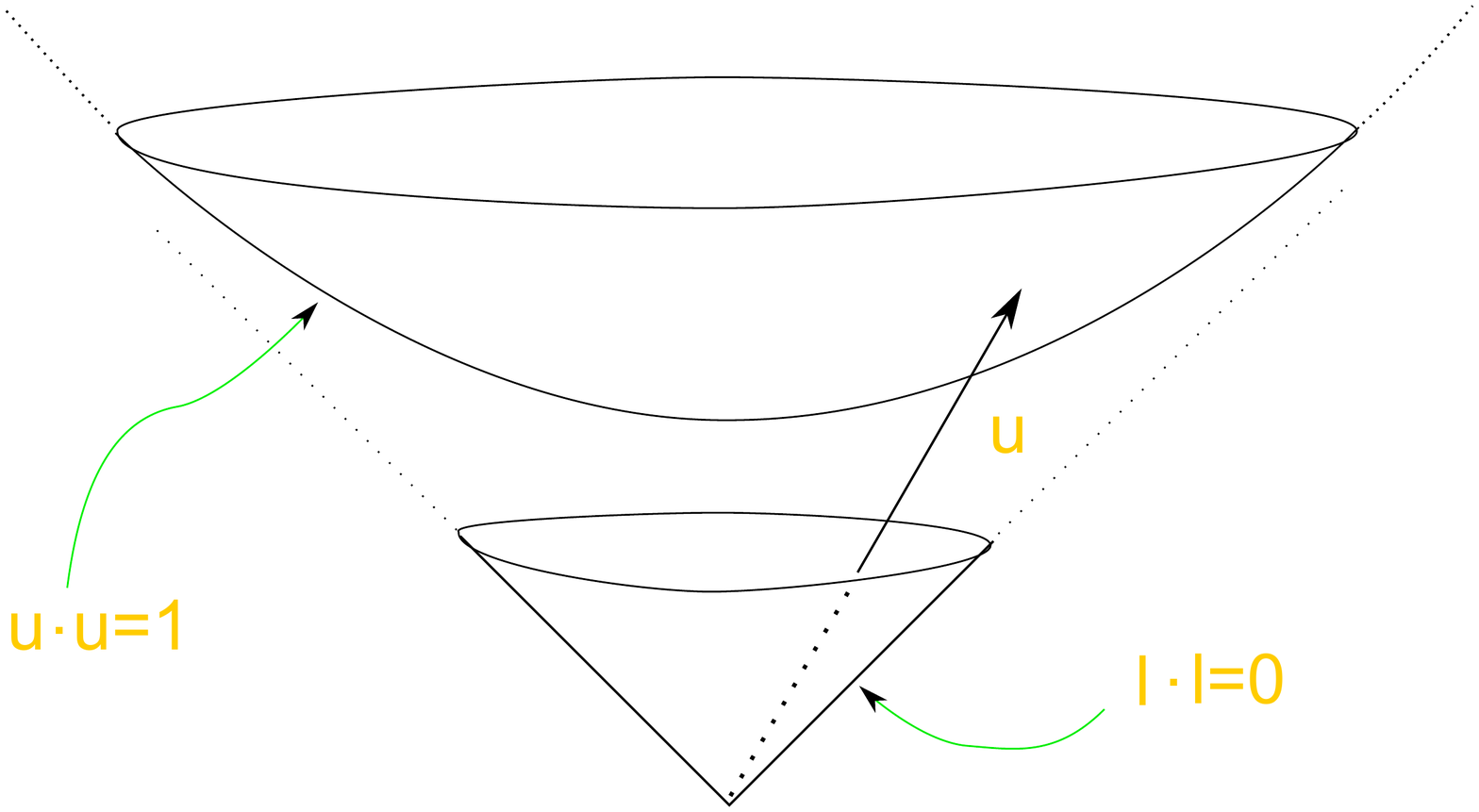}\vspace{-.5cm}\\
Fig.~9. Four-velocities and future-directed lightvectors.
 \eject

 \vspace*{-3cm}
\includegraphics[width=.7\textwidth]{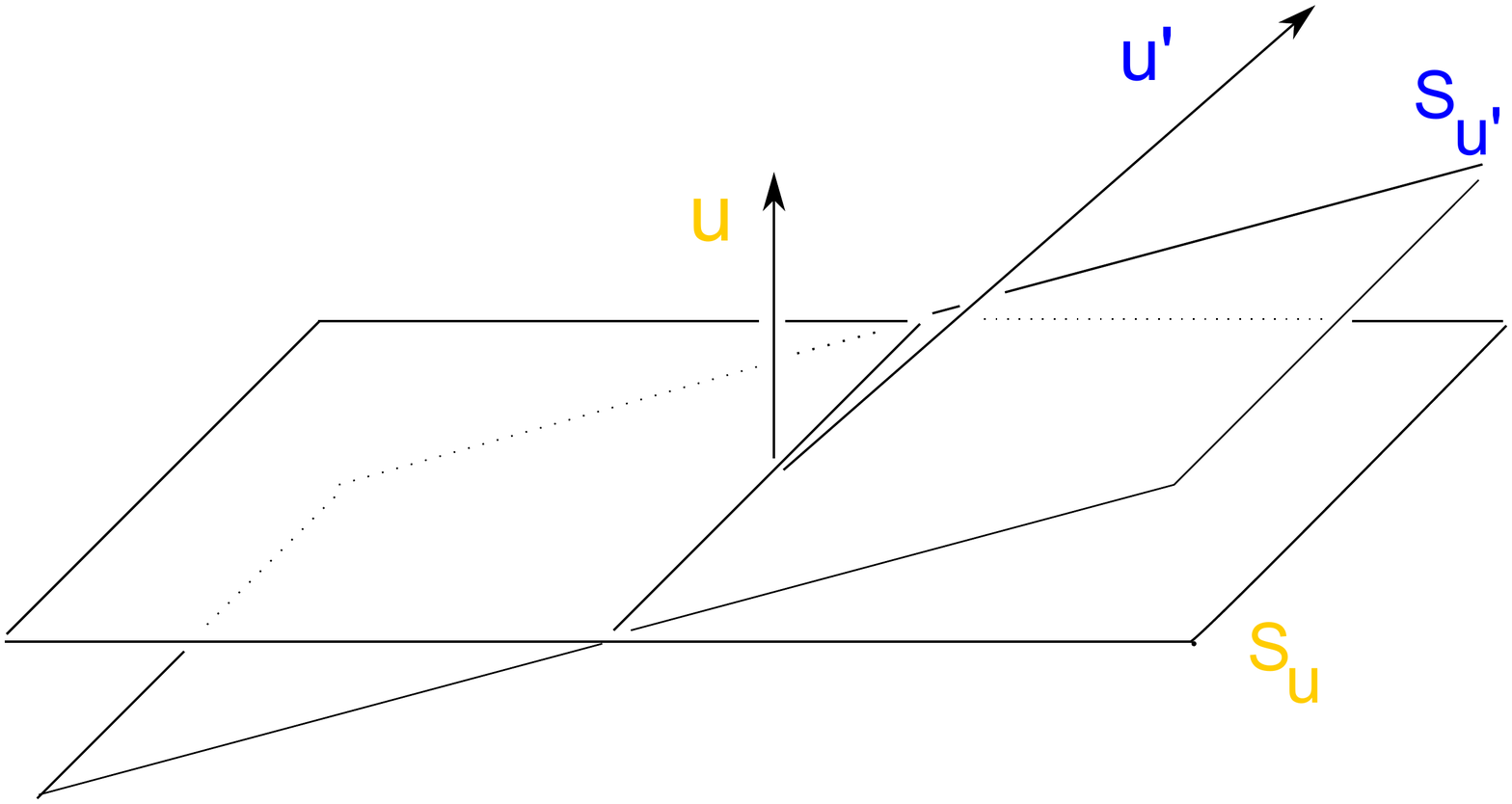}\vspace{-.5cm}\\
Fig.~10. Subspaces orthogonal to 4-velocities.\\
\includegraphics[width=.7\textwidth]{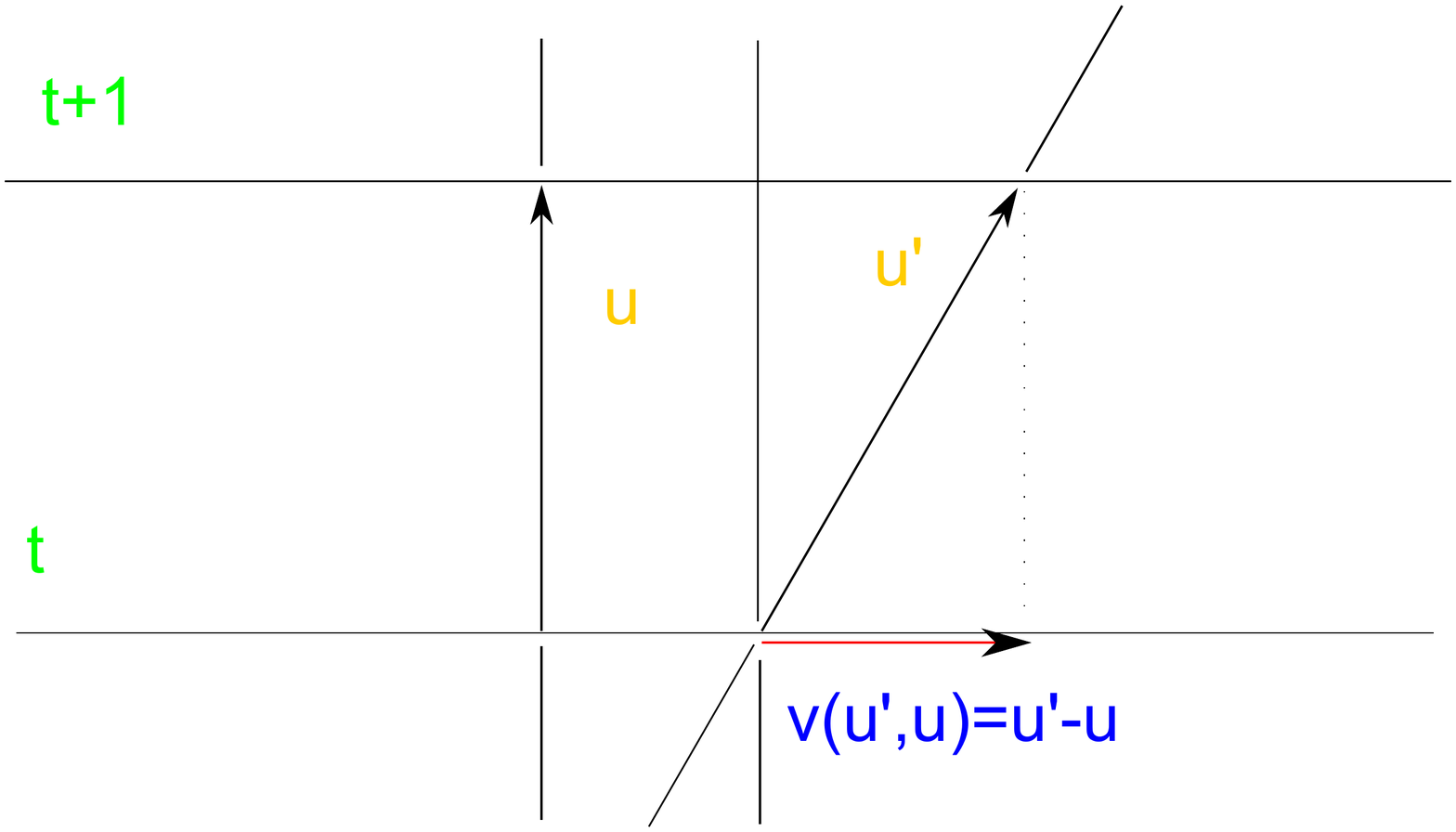}\vspace{-.5cm}\\
Fig.~11. Relative velocity in GS.\\
\includegraphics[width=.7\textwidth]{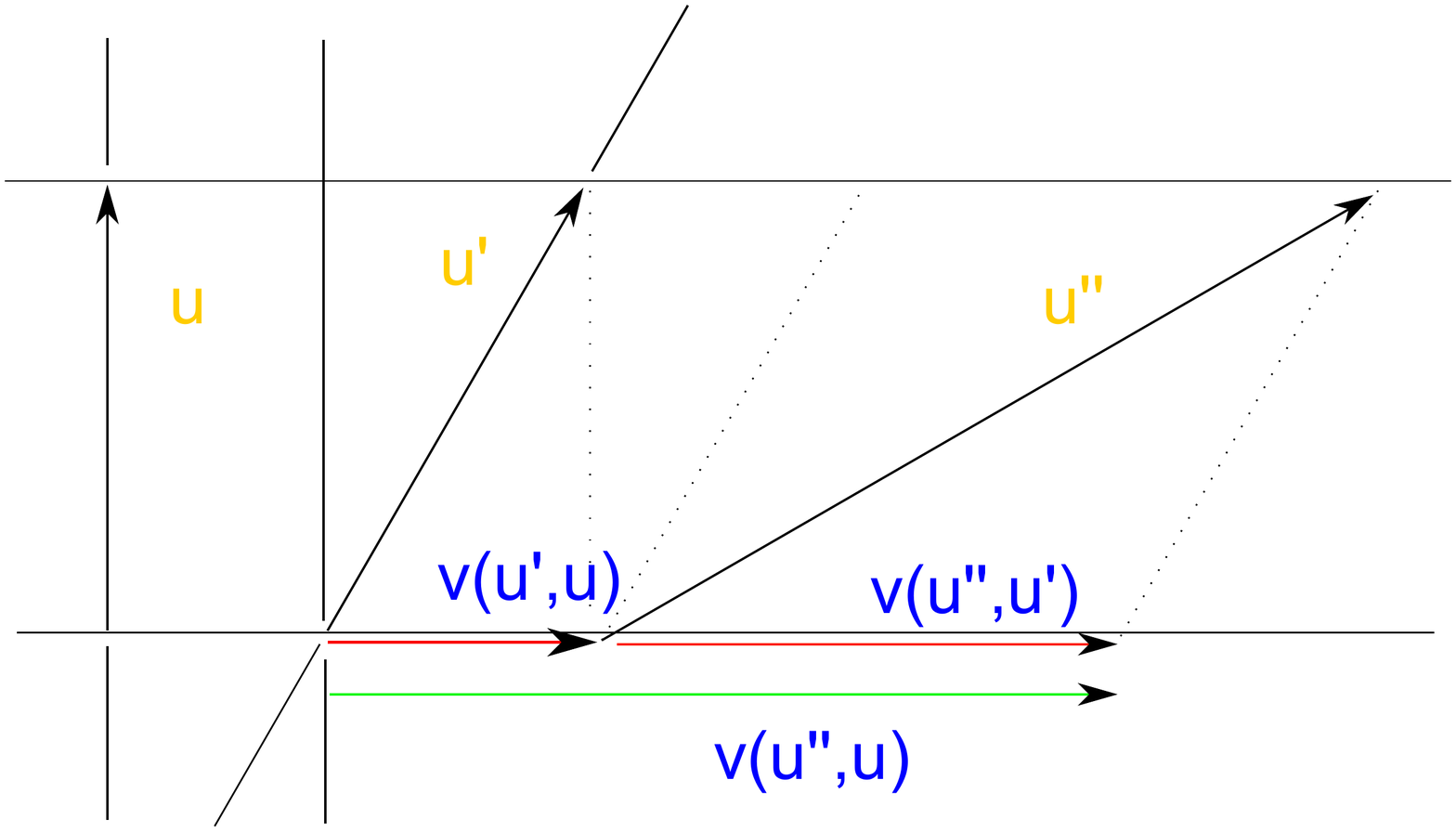}\vspace{-.5cm}\\
Fig.~12. Composition of velocities in GS.
 \eject

 \vspace*{-3cm}
\includegraphics[width=.7\textwidth]{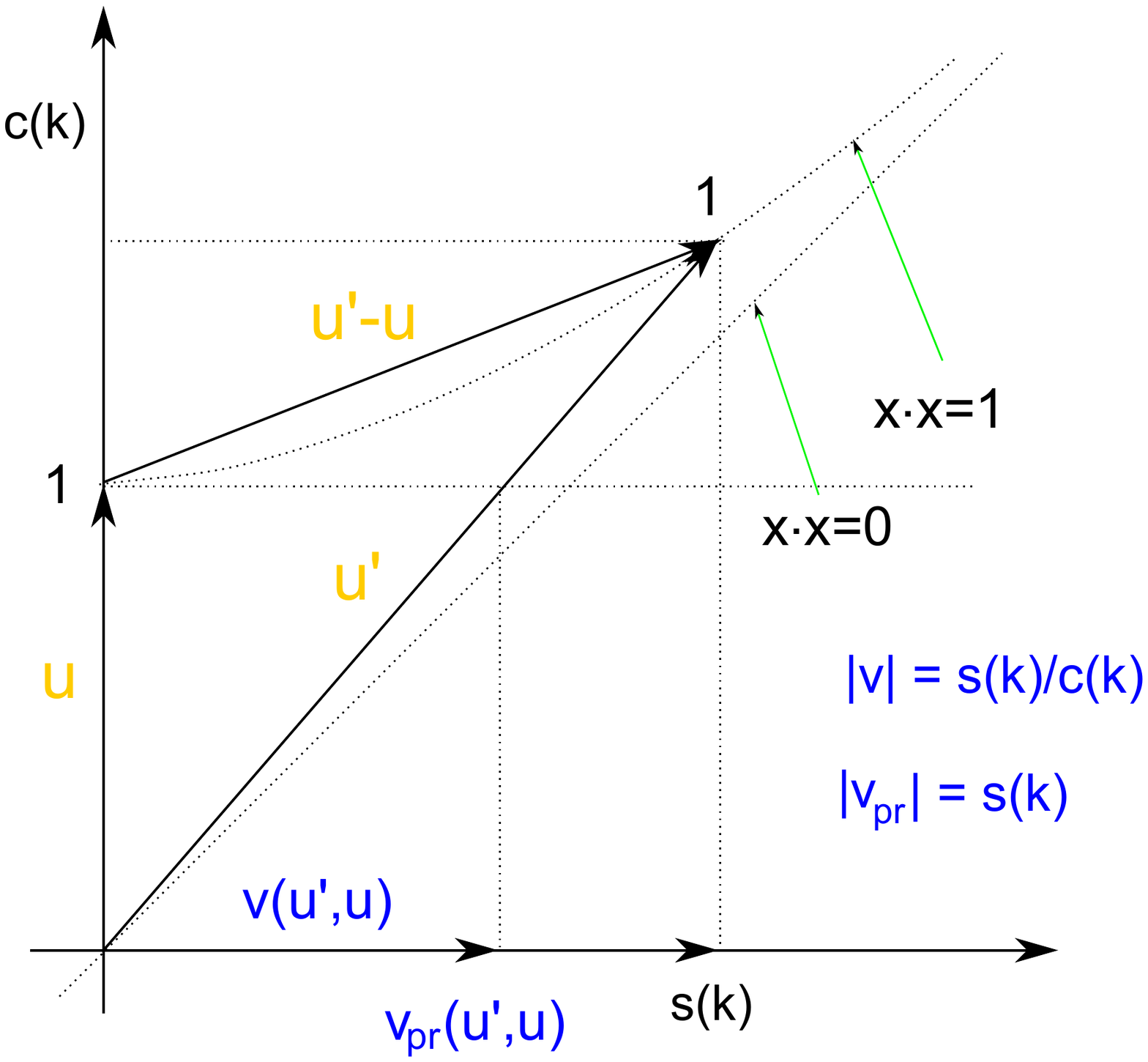}\vspace{-.5cm}\\
Fig.~13. Relative velocity in SR.\\
\includegraphics[width=.7\textwidth]{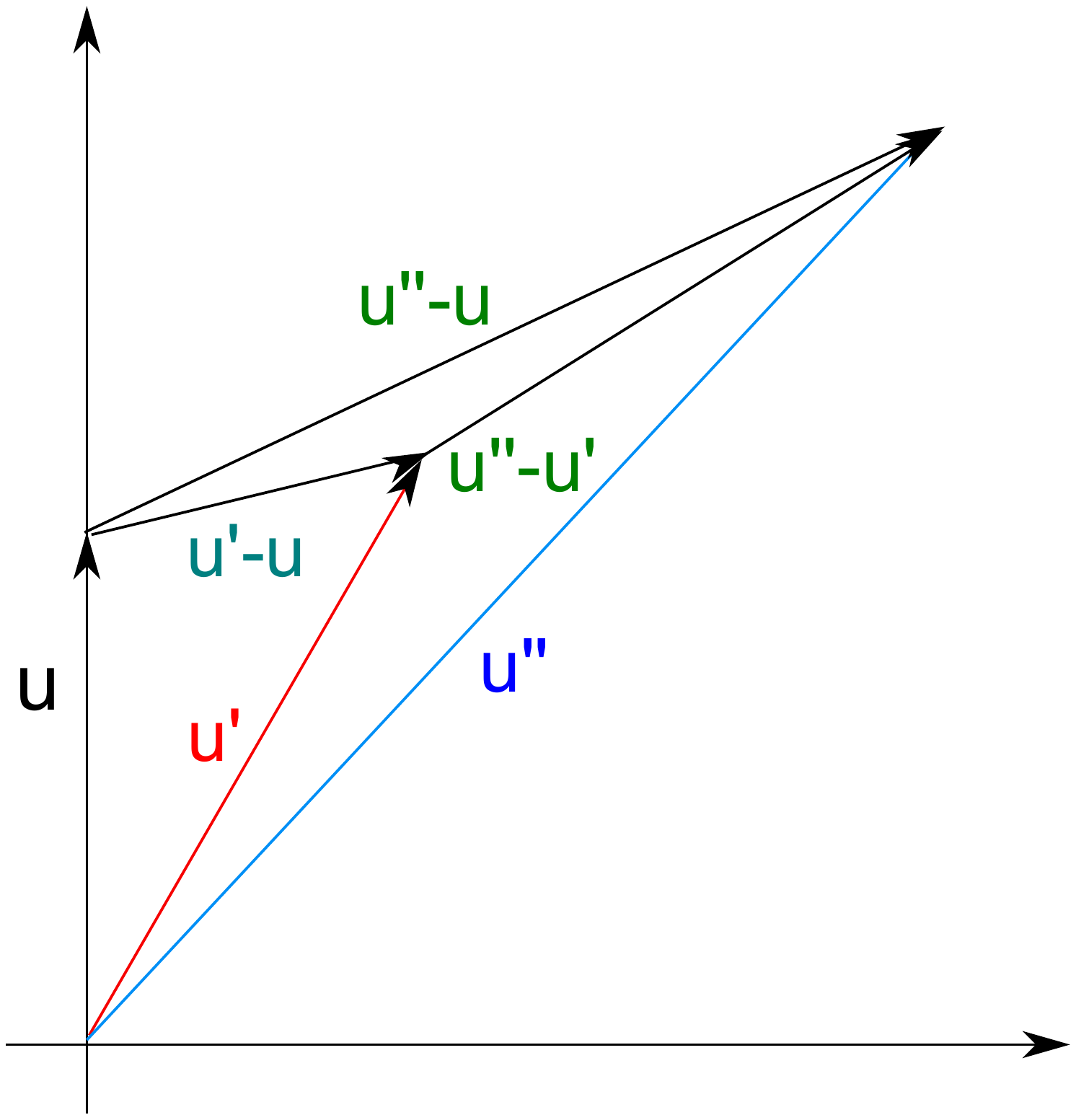}\vspace{-.5cm}\\
Fig.~14. Composition of velocities in SR.\\
\includegraphics[width=.7\textwidth]{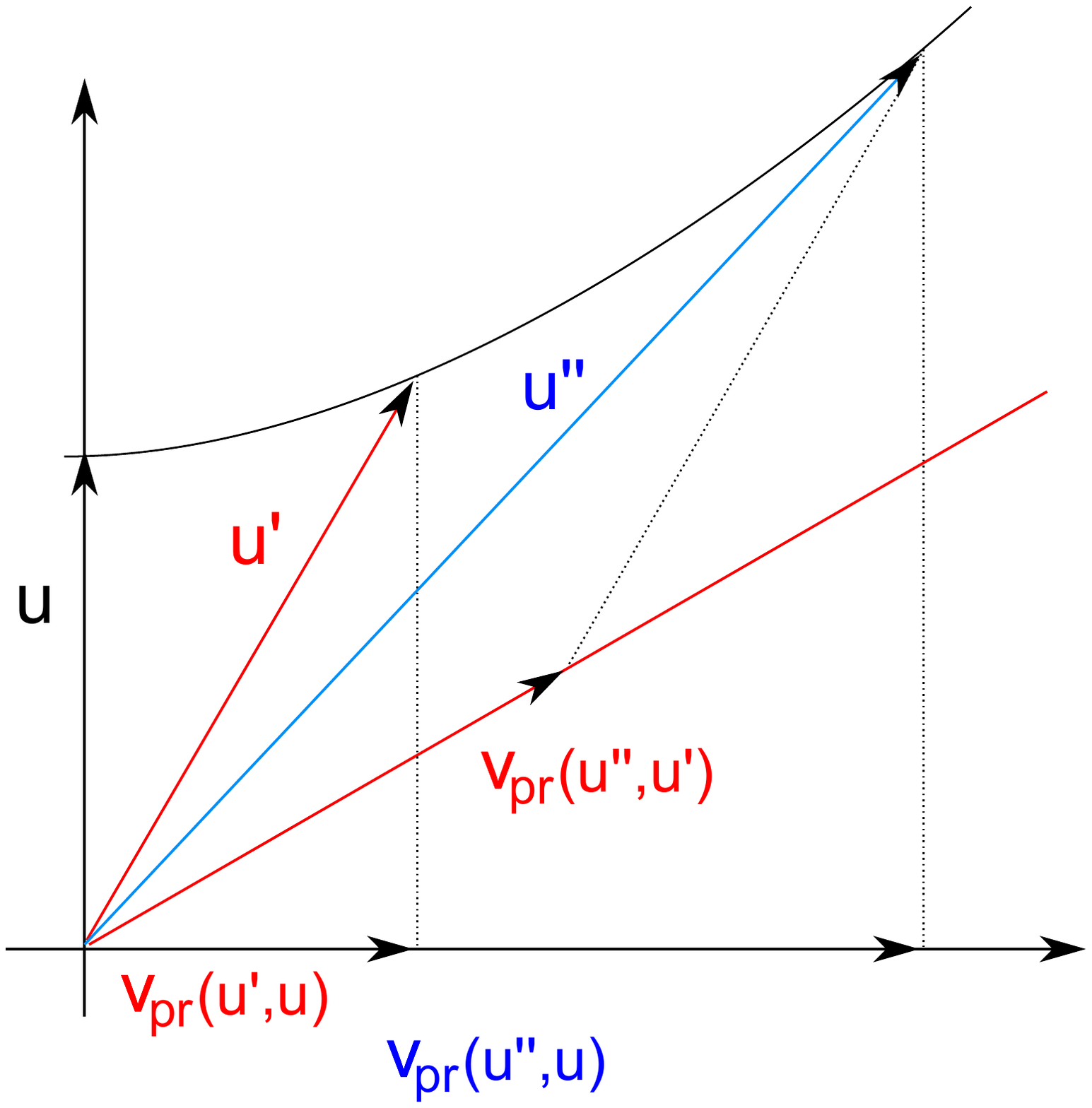}\vspace{-.5cm}\\
Fig.~15. Proper velocities in SR.
 \eject

 \vspace*{-3cm}
\includegraphics[width=.7\textwidth]{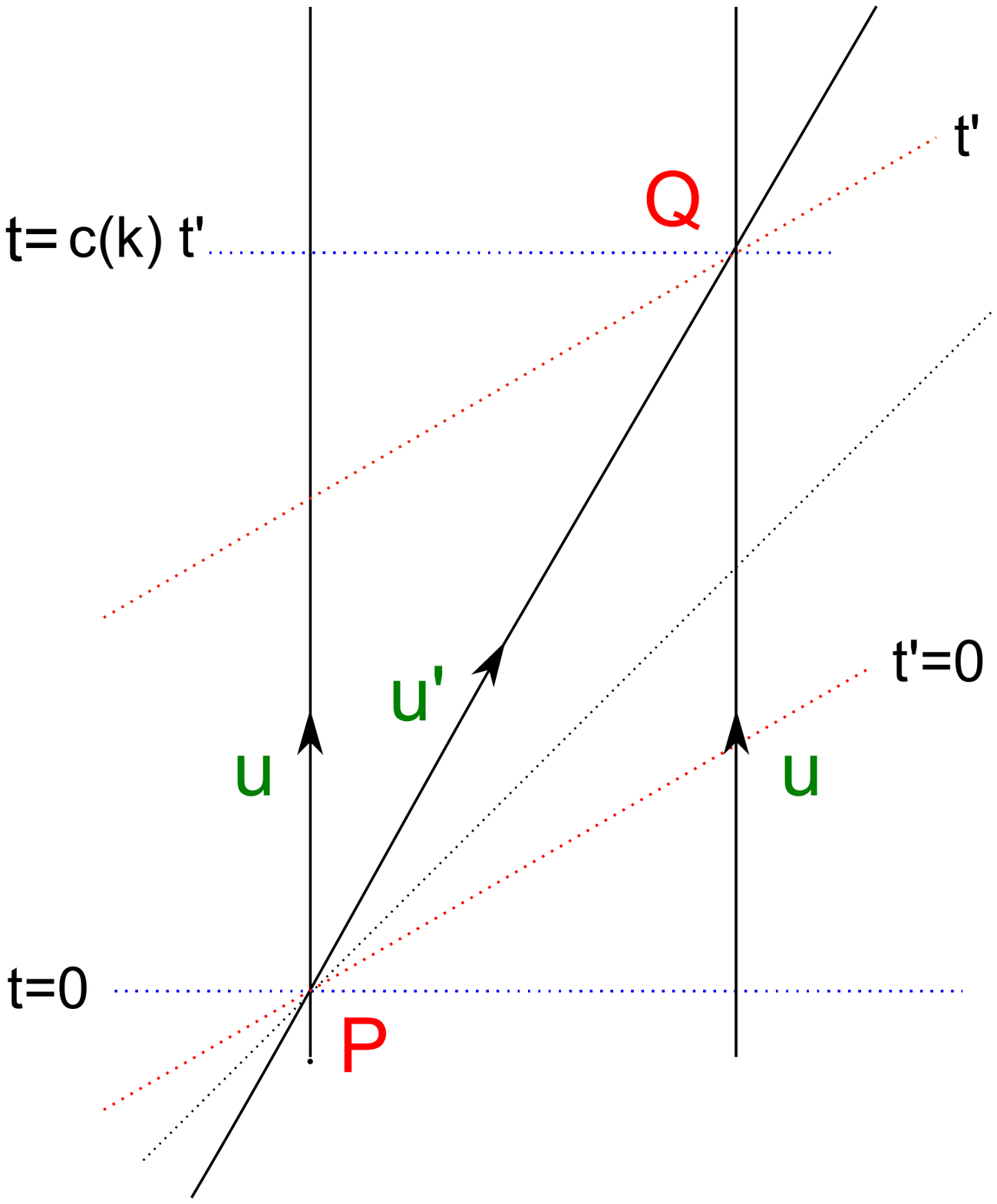}\vspace{-.5cm}\\
Fig.~16. Time measurement.\\
\includegraphics[width=.7\textwidth]{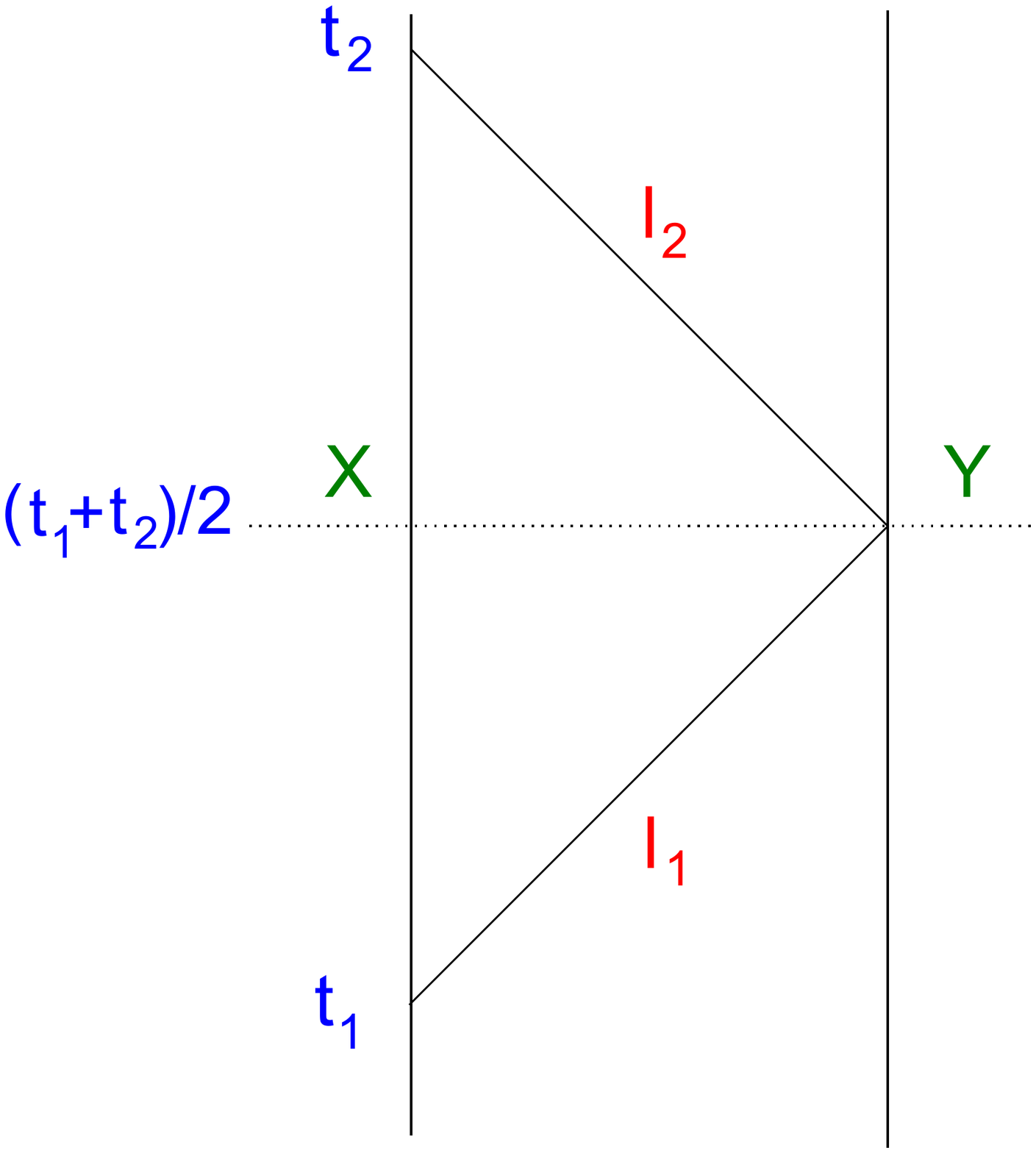}\vspace{-.5cm}\\
Fig.~17. Synchronization of clocks. \\
\includegraphics[width=.7\textwidth]{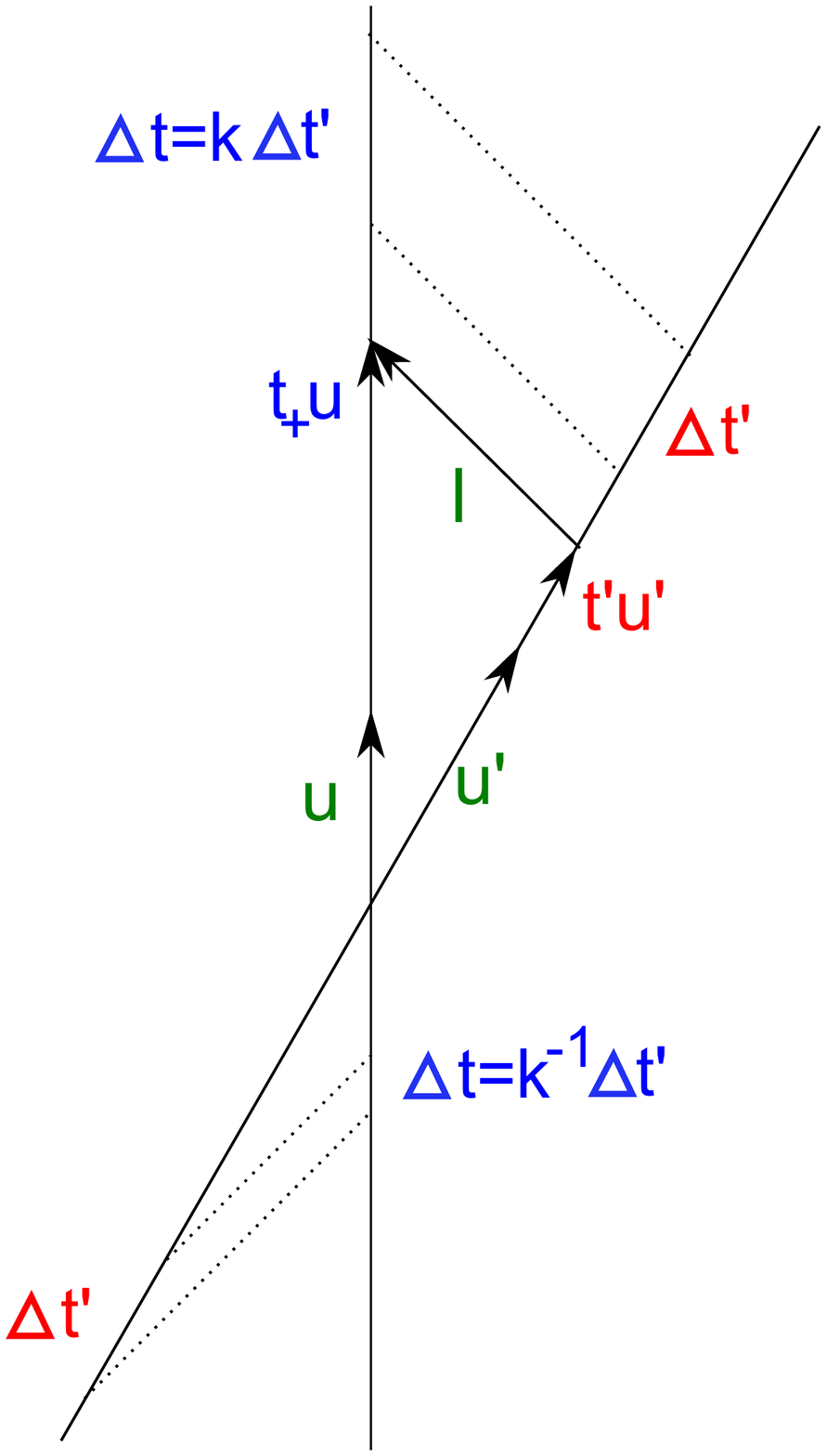}\vspace{-.5cm}\\
Fig.~18. Time of arrival of light signals.

 \eject

 \vspace*{-3.5cm}
\includegraphics[width=.7\textwidth]{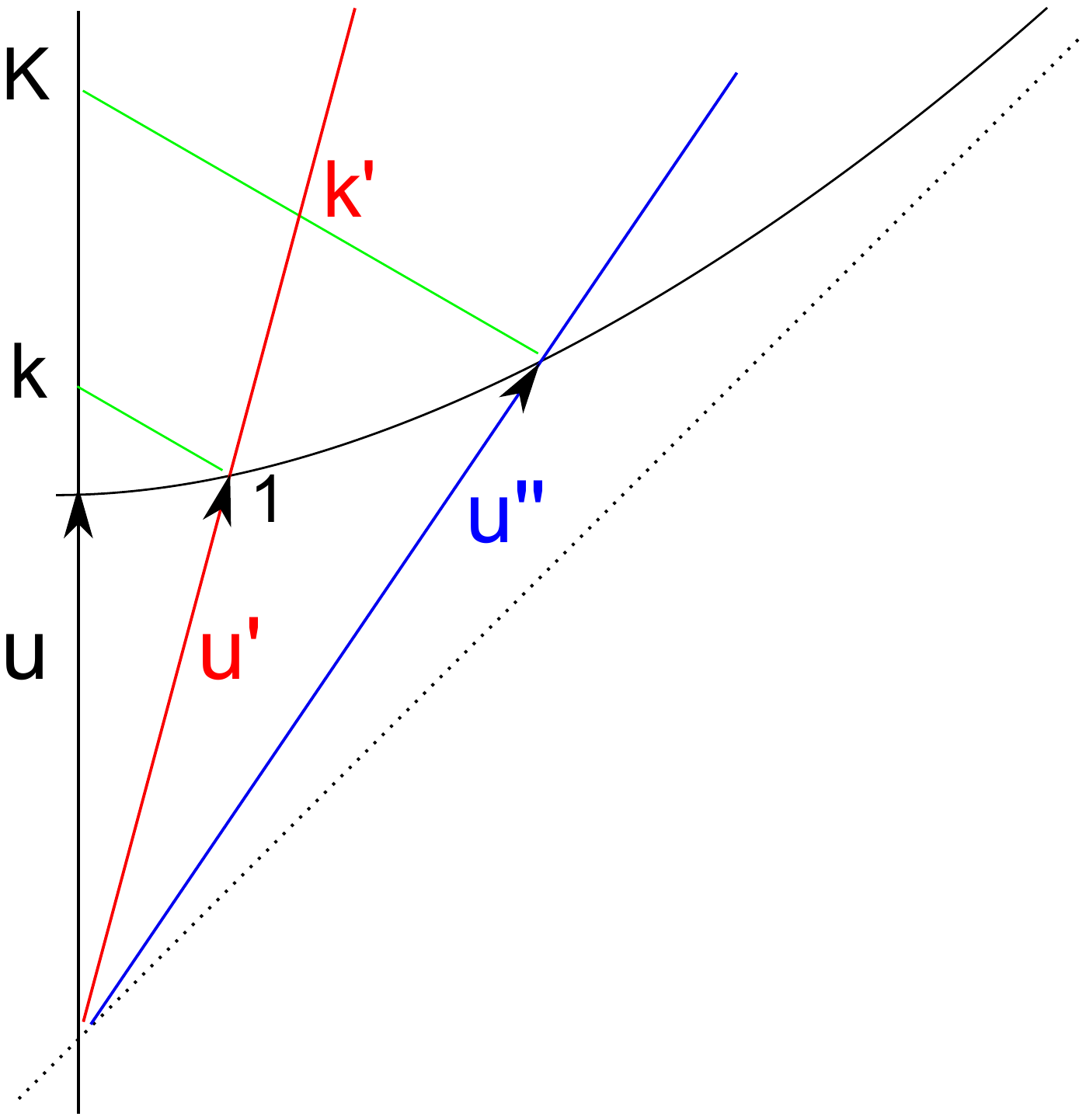}\vspace{-.5cm}\\
Fig.~19. Composition of $k$-coefficients for co-planar
four-velocities: $k/1=K/k'$, so $K=kk'$.\\
\includegraphics[width=.7\textwidth]{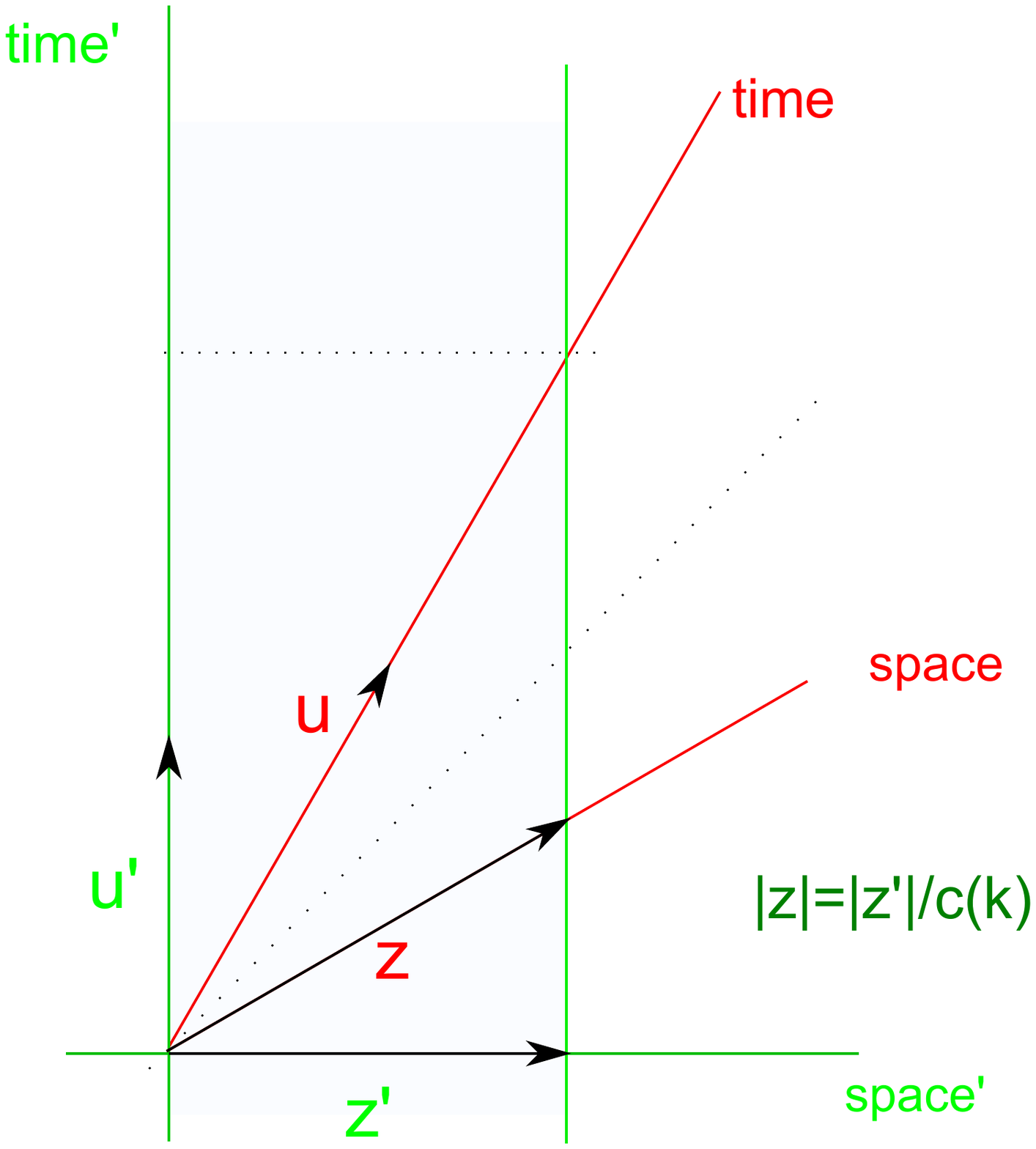}\vspace{-.5cm}\\
Fig.~20. Space measurement.
\includegraphics[width=.7\textwidth]{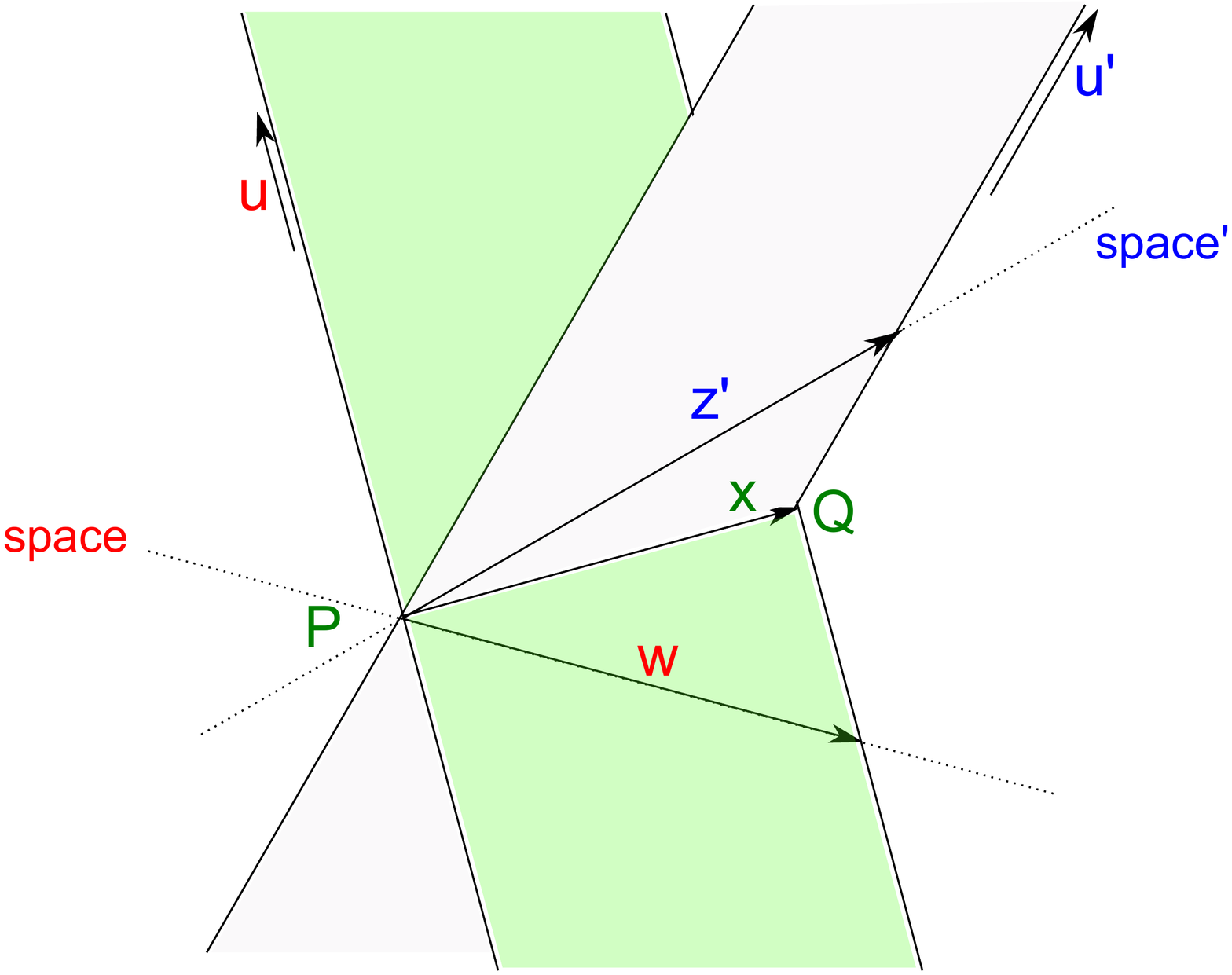}\vspace{-.5cm}\\
Fig.~21. Two rods with ends meeting at $P$ and $Q$
respectively.

 \eject
\vspace*{-3cm}
\includegraphics[width=.7\textwidth]{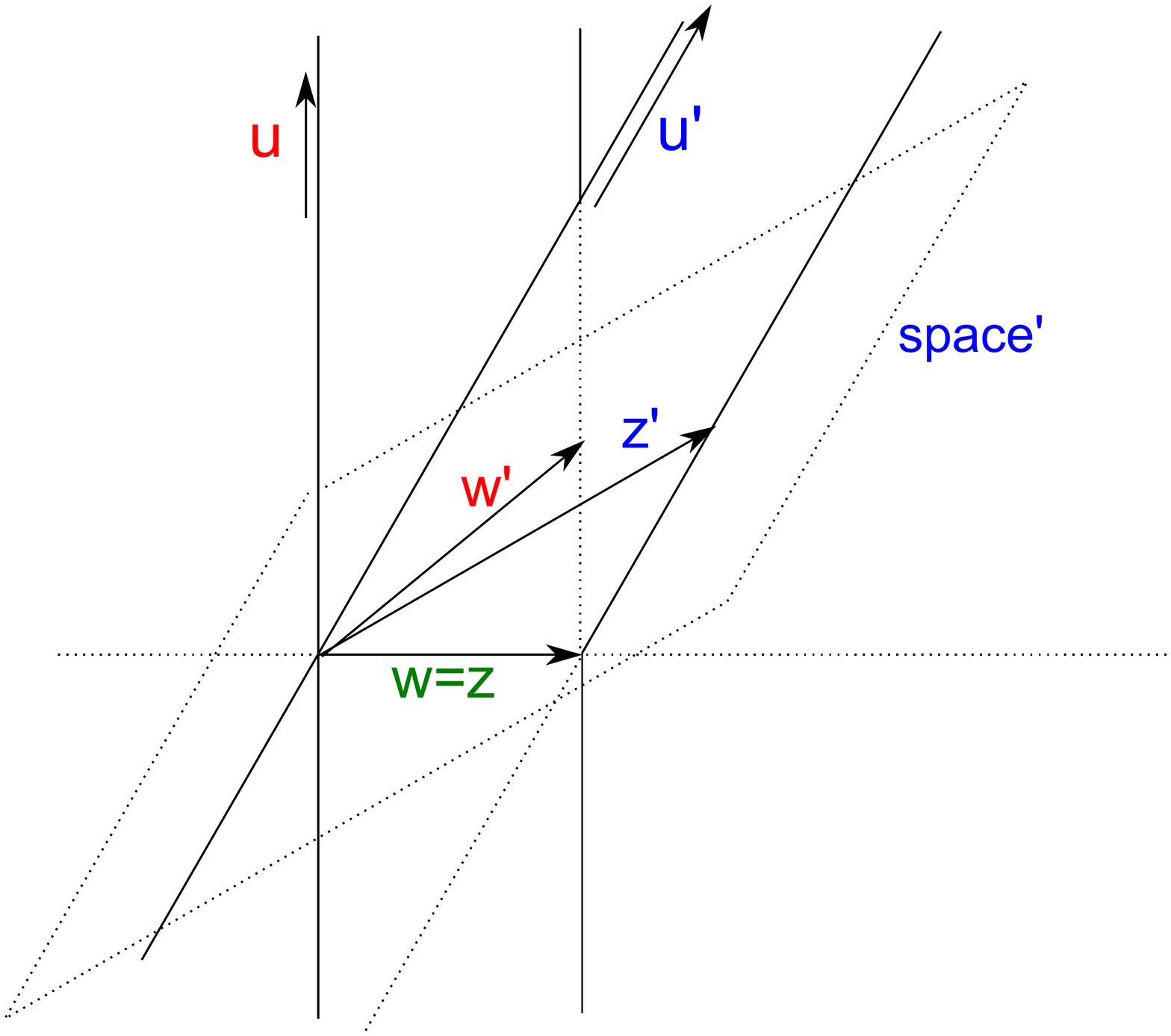}\vspace{-.5cm}\\
Fig.~22. Two rods moving parallelly in $u$-frame, and askew in
$u'$-frame.\\
\includegraphics[width=.7\textwidth]{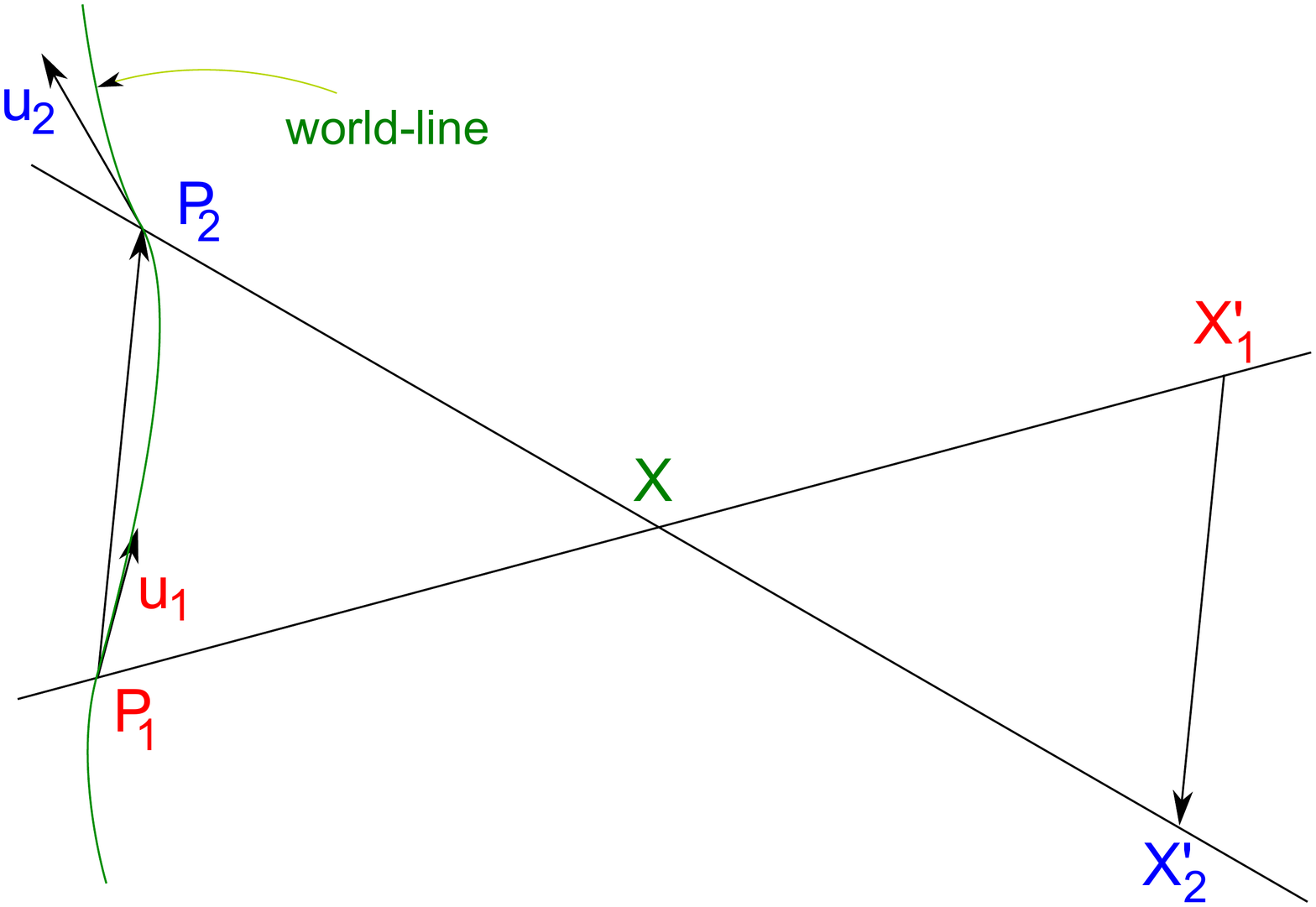}\vspace{-.5cm}\\
Fig.~23. Accelerated motion and simultaneity. \\
\includegraphics[width=.7\textwidth]{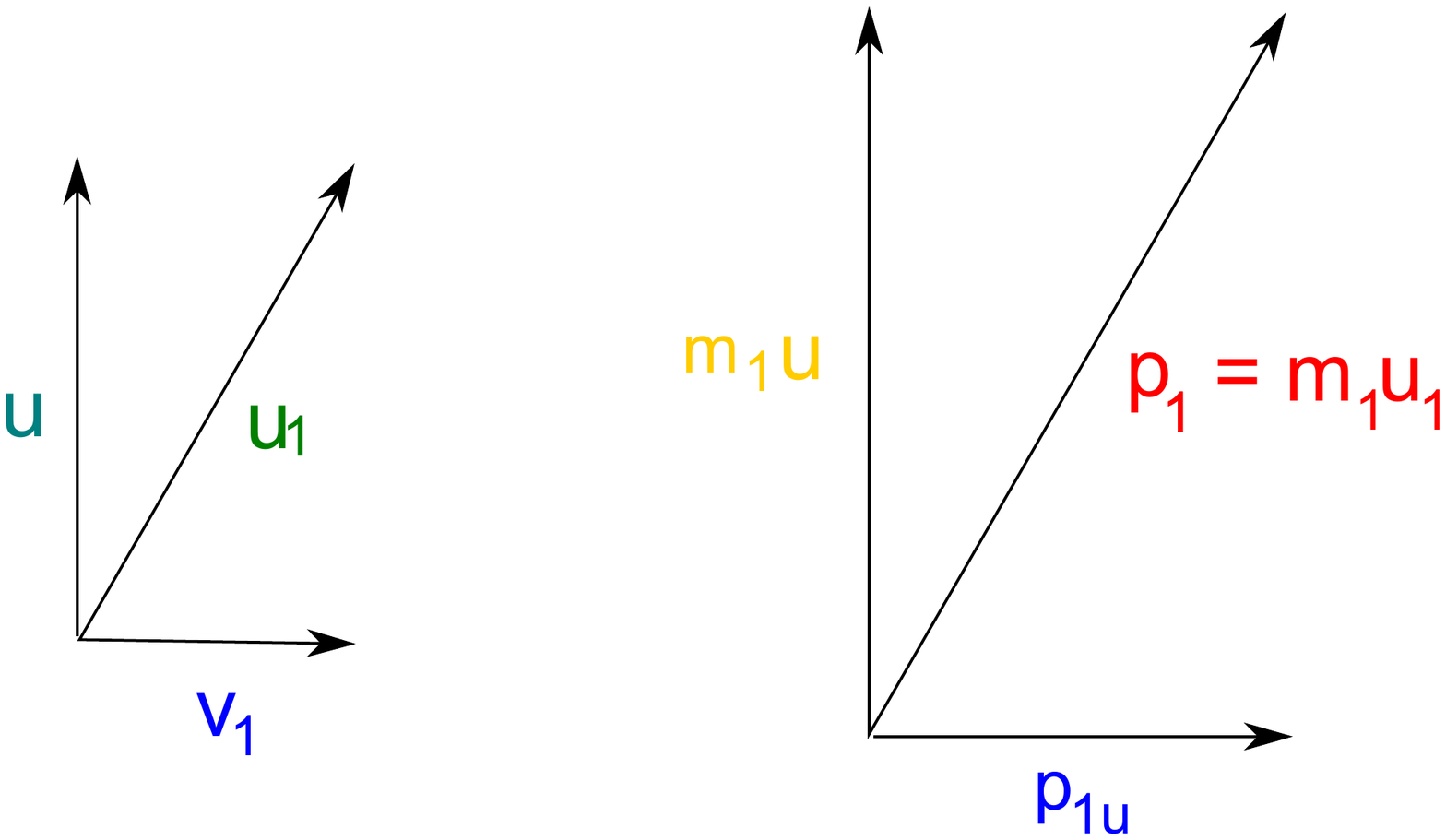}\vspace{-.5cm}\\
Fig.~24. Four-momentum in GS.

 \eject

 \vspace*{-3cm}
\includegraphics[width=.7\textwidth]{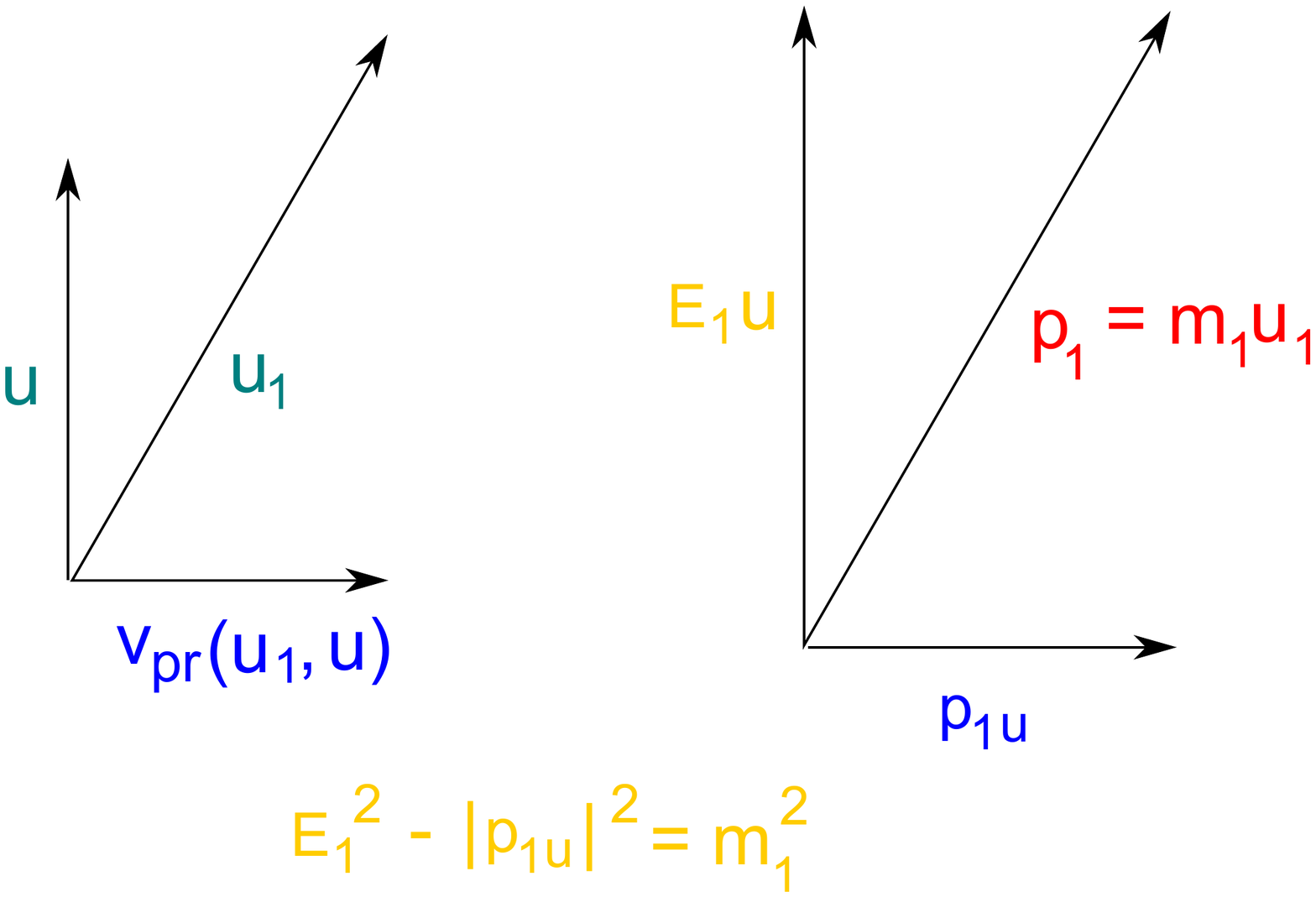}\vspace{-.5cm}\\
Fig.~25. Four-momentum in SR.\\
\includegraphics[width=.7\textwidth]{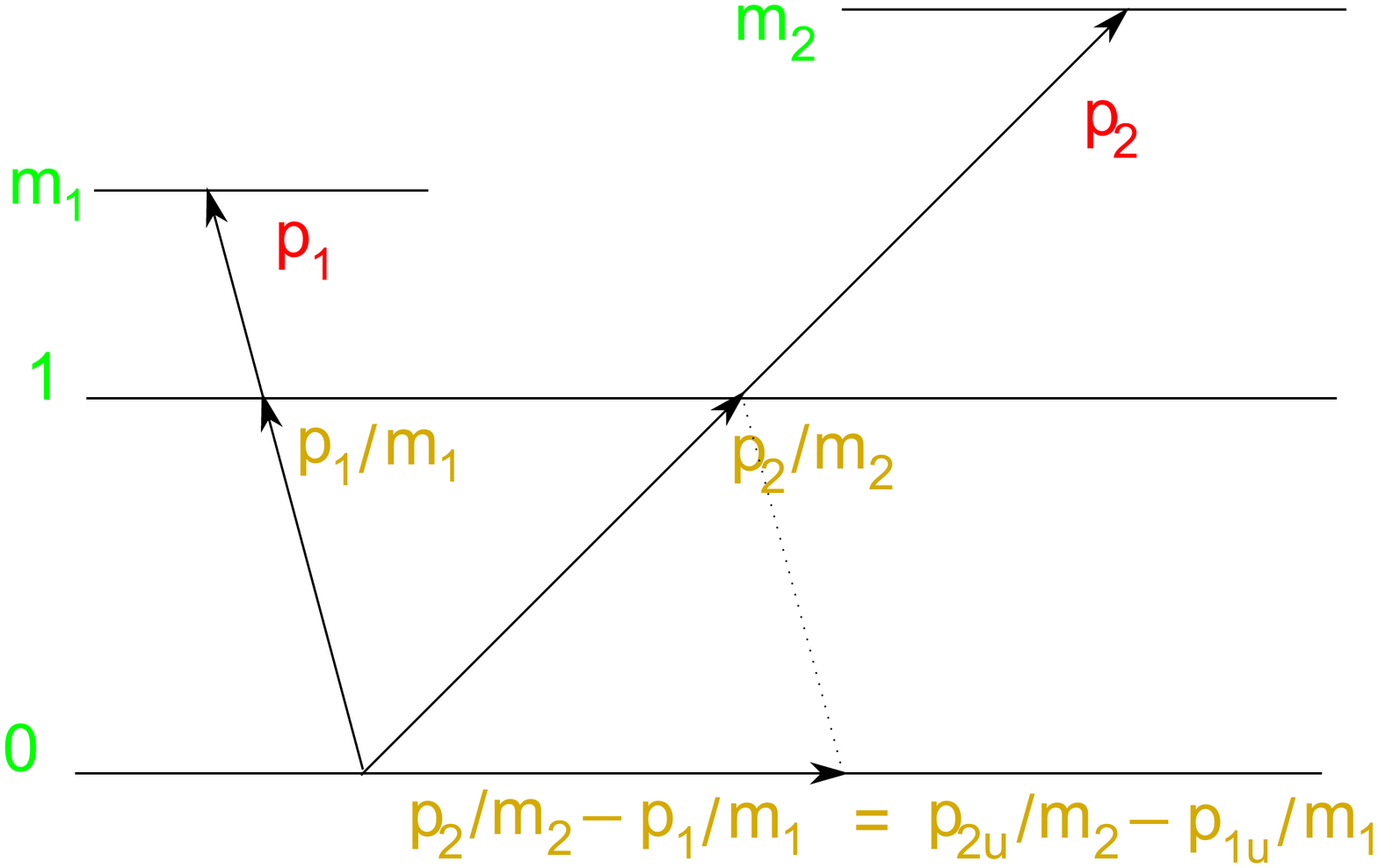}\vspace{-.5cm}\\
Fig.~26. Galilean invariant of two causal vectors:
$|p_{2u}/m_2-p_{1u}/m_1|$\\
\includegraphics[width=.7\textwidth]{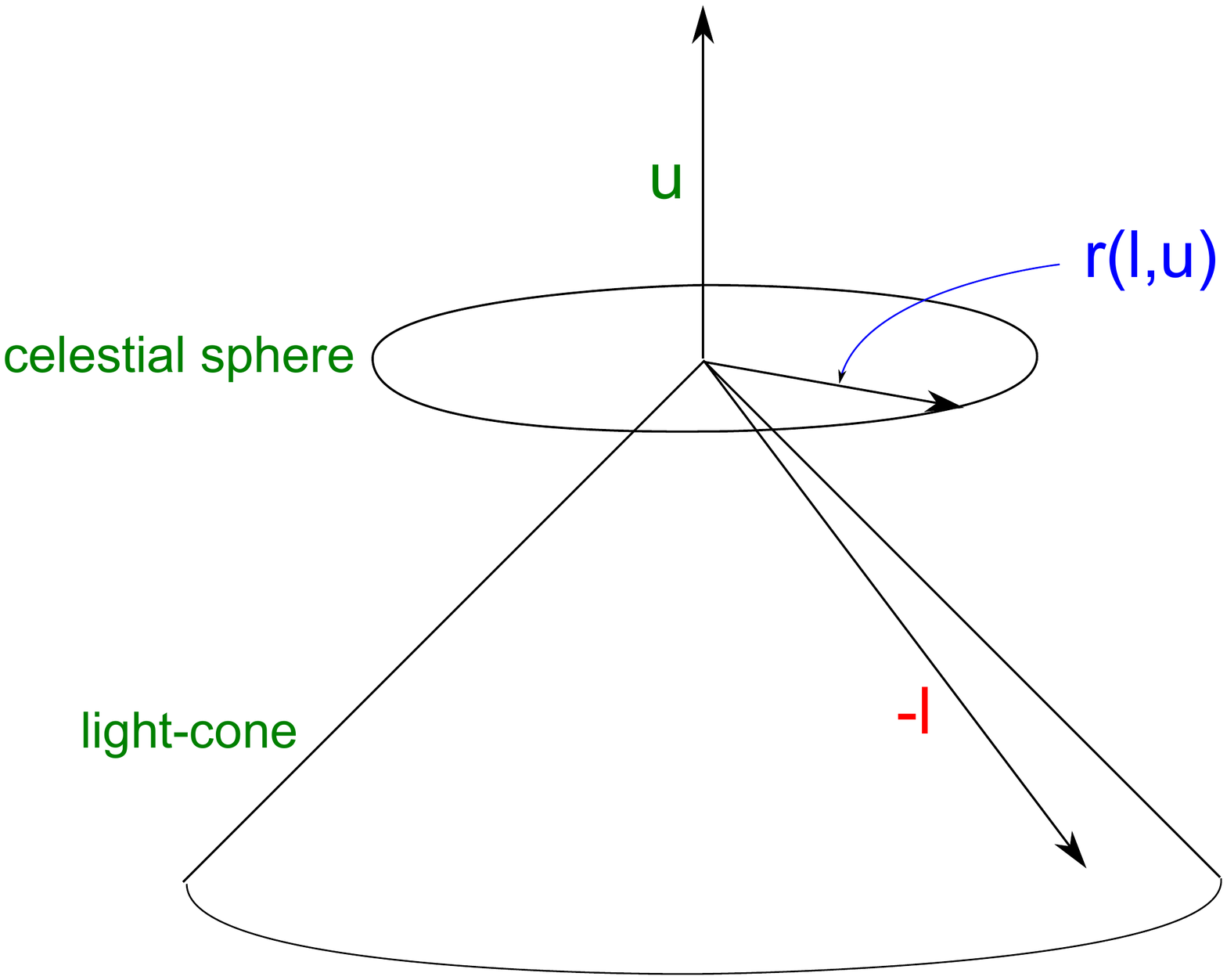}\vspace{-.5cm}\\
Fig.~27. Celestial sphere.

 \eject

 \vspace*{-3cm}
\includegraphics[width=.7\textwidth]{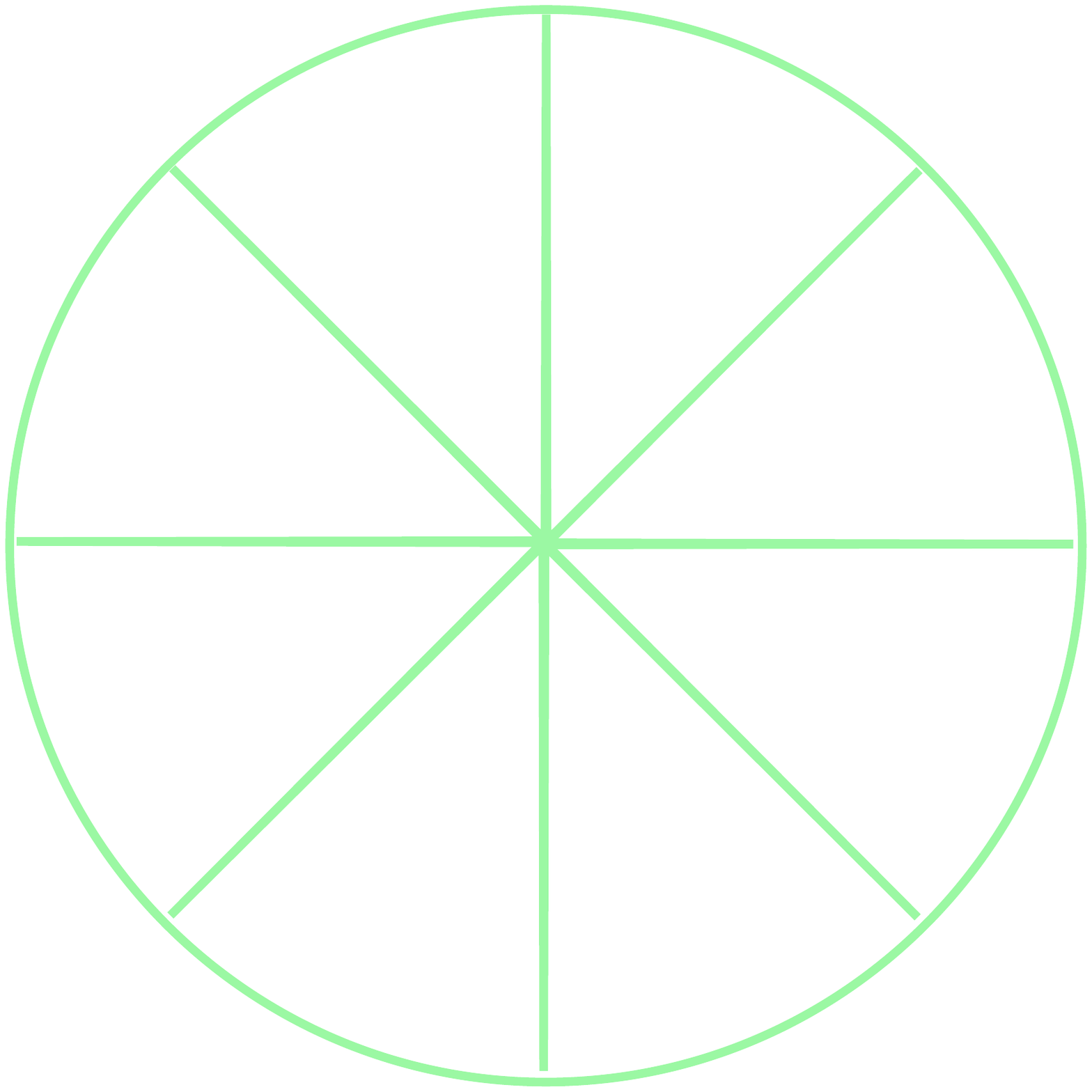}\vspace{-.5cm}\\
Fig.~28. A bicycle wheel in rest.
\includegraphics[width=.7\textwidth]{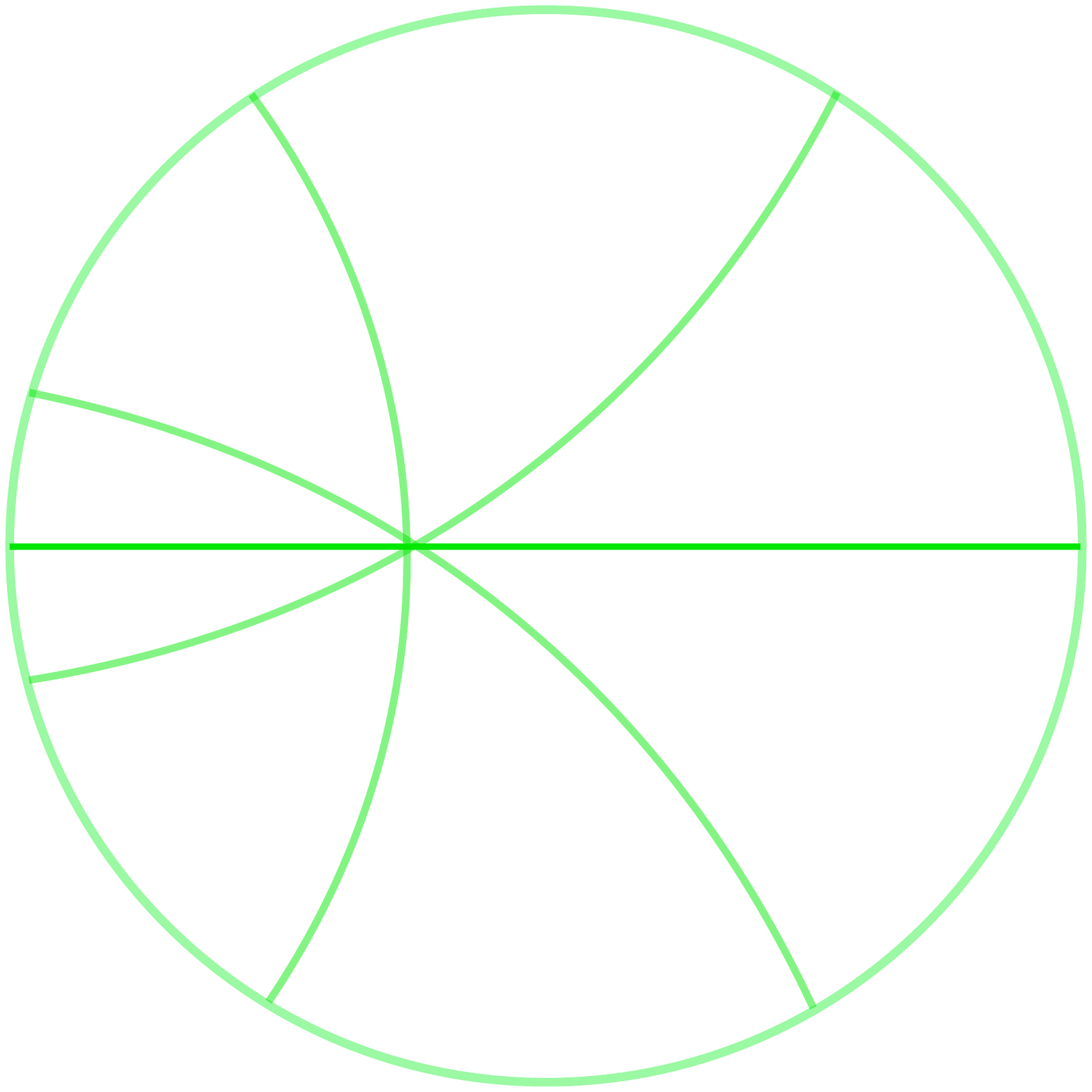}\vspace{-.5cm}\\
Fig.~29. The same wheel as seen by a fast moving observer.
\end{center}


\begin{thebibliography}{99}
\bibitem{cel}However, in the original article
    on the shape of moving sphere, Ref.~4, there is a short
    remark on the idea used in the present article.
\end{thebibliography}
\end{document}